\documentclass{article}
\pdfoutput=1

\usepackage{booktabs,arydshln,tabularray,framed,dashbox}%
\usepackage{subcaption}
\usepackage{xspace}
\usepackage[pdf,all]{xy}
\usepackage{enumitem}
\usepackage{apxproof}
\usepackage{amsmath,amssymb}
\usepackage{authblk}
\usepackage{hyperref}
\usepackage{geometry}
\usepackage{multicol}

\usepackage{algorithm,algorithmicx,algpseudocodex}
\algrenewcommand\algorithmicrequire{\textbf{Input:}}
\algrenewcommand\algorithmicensure{\textbf{Output:}}
\algnewcommand{\Initialize}[1]{\State\textbf{Initialize:} #1}
\algnewcommand\True{\textbf{true}\xspace}
\algnewcommand\False{\textbf{false}\xspace}

\newtheorem{theorem}{Theorem}
\newtheorem{corollary}[theorem]{Corollary}%
\newtheorem{definition}[theorem]{Definition}%
\newtheorem{example}[theorem]{Example}%
\newtheorem{lemma}[theorem]{Lemma}%
\newtheoremrep{lemma}[theorem]{Lemma}%
\newtheorem{proposition}[theorem]{Proposition}%
\newtheoremrep{proposition}[theorem]{Proposition}%

\newenvironment{keywords}
{\par\noindent\textbf{Keywords:}}{\par}

\newcommand{\algand}{\textbf{and}\xspace}

\newcommand{\ab}{\ensuremath{A}\xspace}
\newcommand{\ag}{\text{\normalfont\textsf{Ag}}\xspace}

\newcommand{\CA}{\textnormal{\sf A}\xspace}
\newcommand{\CC}{\textnormal{\sf C}\xspace}
\newcommand{\CE}{\textnormal{\sf E}\xspace}
\newcommand{\CM}{\textnormal{\sf M}\xspace}
\newcommand{\CV}{\textnormal{\sf v}\xspace}
\newcommand{\CW}{\textnormal{\sf W}\xspace}

\newcommand{\lra}{\leftrightarrow}

\newcommand{\NEG}{\ensuremath{\mathord\sim}}
\renewcommand{\phi}{\varphi}
\newcommand{\prop}{\text{\normalfont\textsf{Prop}}\xspace}
\newcommand{\Ra}{\Rightarrow}
\newcommand{\ra}{\rightarrow}

\newcommand{\mbN}{\mathbb{N}}

\newcommand{\mcS}{\mathcal{S}}

\newcommand{\lang}{\ensuremath{\mathcal{EL}}\xspace}
\newcommand{\langc}{\ensuremath{\mathcal{ELC}}\xspace}
\newcommand{\langd}{\ensuremath{\mathcal{ELD}}\xspace}
\newcommand{\langm}{\ensuremath{\mathcal{ELM}}\xspace}
\newcommand{\langcd}{\ensuremath{\mathcal{ELCD}}\xspace}
\newcommand{\langcm}{\ensuremath{\mathcal{ELCM}}\xspace}
\newcommand{\langdm}{\ensuremath{\mathcal{ELDM}}\xspace}
\newcommand{\langcdm}{\ensuremath{\mathcal{ELCDM}}\xspace}
\newcommand{\langl}{\ensuremath{\mathcal{L}}\xspace}
 
\renewcommand{\l}{\text{\normalfont EL}\xspace}
\newcommand{\lc}{\text{\normalfont ELC}\xspace}
\newcommand{\ld}{\text{\normalfont ELD}\xspace}
\newcommand{\lm}{\text{\normalfont ELM}\xspace}
\newcommand{\lcd}{\text{\normalfont ELCD}\xspace}
\newcommand{\lcm}{\text{\normalfont ELCM}\xspace}
\newcommand{\ldm}{\text{\normalfont ELDM}\xspace}
\newcommand{\lcdm}{\text{\normalfont ELCDM}\xspace}

\newcommand{\ls}{\text{\normalfont EL$^s$}\xspace}
\newcommand{\lsc}{\text{\normalfont ELC$^s$}\xspace}
\newcommand{\lsd}{\text{\normalfont ELD$^s$}\xspace}
\newcommand{\lsm}{\text{\normalfont ELM$^s$}\xspace}
\newcommand{\lscd}{\text{\normalfont ELCD$^s$}\xspace}
\newcommand{\lscm}{\text{\normalfont ELCM$^s$}\xspace}
\newcommand{\lsdm}{\text{\normalfont ELDM$^s$}\xspace}
\newcommand{\lscdm}{\text{\normalfont ELCDM$^s$}\xspace}

\newcommand{\K}{\ensuremath{\mathbf{K}}\xspace}
\newcommand{\KC}{\ensuremath{\mathbf{K(C)}}\xspace}
\newcommand{\KD}{\ensuremath{\mathbf{K(D)}}\xspace}
\newcommand{\KCD}{\ensuremath{\mathbf{K(CD)}}\xspace}
\newcommand{\KCDM}{\ensuremath{\mathbf{K(CDM)}}\xspace}
\newcommand{\KCM}{\ensuremath{\mathbf{K(CM)}}\xspace}
\newcommand{\KDM}{\ensuremath{\mathbf{K(DM)}}\xspace}
\newcommand{\KM}{\ensuremath{\mathbf{K(M)}}\xspace}
\newcommand{\KB}{\ensuremath{\mathbf{KB}}\xspace}
\newcommand{\KBC}{\ensuremath{\mathbf{KB(C)}}\xspace}
\newcommand{\KBD}{\ensuremath{\mathbf{KB(D)}}\xspace}
\newcommand{\KBCD}{\ensuremath{\mathbf{KB(CD)}}\xspace}
\newcommand{\KBCDM}{\ensuremath{\mathbf{KB(CDM)}}\xspace}
\newcommand{\KBCM}{\ensuremath{\mathbf{KB(CM)}}\xspace}
\newcommand{\KBDM}{\ensuremath{\mathbf{KB(DM)}}\xspace}
\newcommand{\KBM}{\ensuremath{\mathbf{KB(M)}}\xspace}

\newcommand{\sysC}{\ensuremath{\mathbf{C}}\xspace}
\newcommand{\sysD}{\ensuremath{\mathbf{D}}\xspace}
\newcommand{\sysBD}{\ensuremath{\mathbf{B(D)}}\xspace}
\newcommand{\sysM}{\ensuremath{\mathbf{M}}\xspace}
\newcommand{\sysBM}{\ensuremath{\mathbf{B(M)}}\xspace}

\begin{document}
\title{Epistemic Logic over Similarity Graphs \\{\large Common, Distributed and Mutual Knowledge}}

\date{}

\author[1]{Xiaolong Liang}
\author[2]{Y\`{i} N. W\'{a}ng\footnote{Corresponding author, \href{mailto:ynw@xixilogic.org}{ynw@xixilogic.org}, \href{https://xixilogic.org/ynw}{https://xixilogic.org/ynw}}}
\affil[1]{School of Philosophy, Shanxi University, 92 Wucheng Road, Taiyuan, 030006, Shanxi, P.R. China}
\affil[2]{Department of Philosophy (Zhuhai), Sun Yat-sen University, 2 Daxue Road, Zhuhai, 519082, Guangdong, P.R. China}

\maketitle

\begin{abstract}
In this paper, we delve into the study of epistemic logics, interpreted through similarity models based on weighted graphs. We explore eight languages that extend the traditional epistemic language by incorporating modalities of common, distributed, and mutual knowledge.
The concept of individual knowledge is redefined under these similarity models. It is no longer just a matter of personal knowledge, but is now enriched and understood as knowledge under the individual's epistemic ability. Common knowledge is presented as higher-order knowledge that is universally known to any degree, a definition that aligns with existing literature. We reframe distributed knowledge as a form of knowledge acquired by collectively leveraging the abilities of a group of agents. In contrast, mutual knowledge is defined as the knowledge obtained through the shared abilities of a group.
We then focus on the resulting logics, examining their relative expressivity, semantic correspondence to the classical epistemic logic, proof systems and the computational complexity associated with the model checking problem and the satisfiability/validity problem. This paper offers significant insights into the logical analysis and understanding of these enriched forms of knowledge, contributing to the broader discourse on epistemic logic.\vspace{0.5em}
\end{abstract}

\section{Introduction}

Even though the concept of \emph{similarity} is intrinsically linked to \emph{knowledge}, it has not been traditionally emphasized or explicitly incorporated in the classical representation of knowledge within the field of epistemic logic \cite{Hintikka1962,FHMV1995,MvdH1995}. This could be due to a variety of reasons including the complexity of quantifying similarity or the traditional focus on other aspects of knowledge representation. However, over recent years, researchers have started to probe this relationship more deeply, marking a fresh direction in the field \cite{DLW2021,NT2015}. This is indicative of an evolving understanding and appreciation of the role that similarity can play in shaping and defining knowledge structures.

The technical framework for exploring this relationship has its roots in weighted modal logics \cite{LM2014,HLMP2018}. This approach offers a quantitative way of considering similarity, allowing for a more nuanced understanding of knowledge. For example, it can be used to model the concept that some pieces of knowledge are more similar or relate more closely to one another than others. Our work distinguishes itself from recent advancements in epistemic logic interpreted through the concepts of similarity or distance (e.g., \cite{DLW2021,NT2015}).  One key difference is that we employ the standard language of epistemic logic, whereas other solutions typically incorporate the degree of similarity or dissimilarity directly into their language. For instance, they use a sentence like $K_a^r \phi$ to represent ``$a$ knows $\phi$ with a strength of effort or confidence $r$''. Our approach, in contrast, does not explicitly factor in the degree of similarity into the language, maintaining the traditional structure of epistemic logic while reinterpreting its concepts in the light of similarity.

In this paper, we adapt the concept of similarity from the field of data mining, where it is primarily used to quantify the likeness between two data objects. In data mining, distance and similarity measures are generally specific algorithms tailored to particular scenarios, such as computing the distance and similarity between matrices, texts, graphs, etc. (see, e.g., \cite[Chapter~3]{Aggarwal2015}). There is also a body of literature that outlines general properties of distance and similarity measures. For instance, in \cite{TSK2005}, it is suggested that typically, the properties of \emph{positivity} (i.e., $\forall x \forall y: s(x, y) = 1 \Ra x = y$) and \emph{symmetry} (i.e., $\forall x \forall y: s(x, y) = s(y, x)$) hold for $s(x, y)$ -- a binary numerical function that maps the similarity between points $x$ and $y$ to the range $[0, 1]$.

However, our primary interest does not lie in the measures of similarity themselves, but rather in modeling similarity and deriving from it the concepts of knowledge. We accomplish this by interpreting knowledge through a category of models (and a specific subtype known as similarity models). Intuitively, the phrase ``$a$ knows $\phi$'' ($K_a\phi$) can be interpreted as ``$\phi$ holds true in all states that, in $a$'s perception, resemble the actual state''. A ``state'' in this context can be seen as a data object, which is the focus of data mining. But it could also be treated as an epistemic object, a possible situation, and so forth.

We generalize the similarity function by replacing its range [0, 1] with an arbitrary set of epistemic abilities. In our framework, the degrees of similarity may not have a comparable or ordered relationship. This shift allows for a more nuanced understanding and modeling of knowledge, accommodating the complex and often non-linear nature of how knowledge is represented and interpreted. This approach has the potential to provide a more flexible framework for representing knowledge in various fields where the traditional binary or linear models of knowledge fall short. It may be particularly relevant in areas like artificial intelligence, cognitive science, and social sciences where the understanding and modeling of knowledge need to take into account complex human cognition and social dynamics.

The primary focus of this paper is on \emph{group knowledge}. The flexibility of our models allows us not only to reinterpret classical concepts like \emph{everyone's knowledge} ($E_G\phi$), \emph{common knowledge} ($C_G\phi$), and \emph{distributed knowledge} ($D_G\phi$) while preserving their underlying intuitions, but it also leads to the introduction of a uniquely natural and novel concept of group knowledge that we term as \emph{mutual knowledge} ($M_G\phi$). This new concept is conceptually similar to everyone's knowledge, though it bears subtle technical differences. In simple terms, everyone's knowledge and common knowledge maintain their standard relationship to individual knowledge (though, in our models, individual knowledge is reinterpreted with an emphasis on the role of similarities). Distributed knowledge, on the other hand, is not defined by intersections of relations in our approach. Instead, we employ unions of epistemic abilities, signifying knowledge that can be acquired by combining the abilities of agents within the group. The newly introduced concept of mutual knowledge represents knowledge that can be gained through the mutual abilities of agents within the group. This is akin to everyone's knowledge yet it brings a fresh perspective by stressing on the collective abilities of the group.

In our study, we explore epistemic logics across all combinations of these group knowledge notions. As everyone's knowledge can be expressed by individual knowledge, we have formulated eight languages (with or without common, distributed, and mutual knowledge) and sixteen logics over these languages. Each language is interpreted over either the class of models or the class of similarity models. The grammar and semantics of these languages are introduced in Sections~\ref{sec:syntax}--\ref{sec:semantics}.
In Section~\ref{sec:expressivity} , we compare the expressive power of these languages. We establish correspondence, as discussed in Section~\ref{sec:correspondence}, between these logics (excluding mutual knowledge) and classical ones. This correlation is beneficial for accomplishing some of the results related to axiomatization and computational complexity, which are then the subjects of Sections~\ref{sec:ax} and \ref{sec:complexity} respectively.
For the axiomatizations of the logics, we introduce sound and strongly complete axiomatic systems for the logics excluding common knowledge. For those incorporating common knowledge, we present sound and weakly complete axiomatic systems (owing to the lack of compactness for the common knowledge operators). These systems are then categorized based on whether their completeness results are obtainable via correspondence (Section~\ref{sec:completeness1}), shown via a path-based canonical model (Section~\ref{sec:completeness2}), or require a finitary method leading to a weak completeness result (Section~\ref{sec:completeness3}).
In terms of computational complexity, we initially explore the model checking problems -- all of which are in P (Section~\ref{sec:mc}). Subsequently, we investigate the satisfiability/validity problems -- those without common knowledge are PSPACE complete, while the others are EXPTIME complete (Section~\ref{sec:sat}).

Our findings demonstrate that models built on weighted graphs offer a more subtle and adaptable framework for knowledge modeling. This is especially true when viewed from the perspective of a similarity measure (or its dual, a distance measure, which can be implemented in a very similar manner) and its connection to the notions of knowledge. In comparison to classical epistemic logics, our logics do not exhibit greater complexity. This leads to an improved balance between the characterization power of the logics and their computational complexity. Here, weighted models reveal a distinct advantage! Our exploration of the relationship between similarity and knowledge could potentially narrow the gap between research on epistemic logic and the various directions related to knowledge representation and beyond. Despite the limited attention it has received in the past, the concept of similarity is gradually gaining recognition as a vital factor in the study of knowledge representation in epistemic logic. This underlines the necessity for ongoing research and investigation in this area.

\section{Syntax and Semantics}
\label{sec:logics}

In this section, we present a comprehensive framework composed of eight formal languages, two types of formal models, and a unified semantic interpretation. The combination of these elements results in a diverse collection of sixteen distinctive logics. We supplement our discussion with illustrative examples, offering a visual representation of the models and their accompanying semantics.

\subsection{Syntax}\label{sec:syntax}

Our study utilizes formal languages rooted in the standard language of multi-agent epistemic logic \cite{FHMV1995,MvdH1995}, with the addition of modalities that represent group knowledge constructs. We particularly concentrate on the constructs of \emph{common knowledge}, \emph{distributed knowledge} and \emph{mutual knowledge}. Although these terms may be familiar within the realm of epistemic logic, their interpretations within our framework will be uniquely defined and explicated later in Section~\ref{sec:semantics}.

In terms of our assumptions, we consider \prop as a countably infinite set of propositional variables, and \ag as a finite nonempty set of agents. With these in place, we now proceed to delineate the formal languages.

\begin{definition}[formal languages]
The languages utilized in our study are defined by the following grammatical rules, where the name of each language is indicated in parentheses on the left-hand side:
$$\begin{array}{l@{\qquad}l}
(\lang)&\phi ::= p \mid \neg \phi \mid (\phi \ra \phi) \mid K_a\phi\\
(\langc)&\phi ::= p \mid \neg \phi \mid (\phi \ra \phi) \mid K_a\phi \mid C_G\phi\\
(\langd)&\phi ::= p \mid \neg \phi \mid (\phi \ra \phi) \mid K_a\phi \mid D_G\phi\\
(\langm)&\phi ::= p \mid \neg \phi \mid (\phi \ra \phi) \mid K_a\phi \mid M_G\phi\\
(\langcd)&\phi ::= p \mid \neg \phi \mid (\phi \ra \phi) \mid K_a\phi \mid C_G\phi \mid D_G\phi \\
(\langcm)&\phi ::= p \mid \neg \phi \mid (\phi \ra \phi) \mid K_a\phi \mid C_G\phi \mid M_G\phi \\
(\langdm)&\phi ::= p \mid \neg \phi \mid (\phi \ra \phi) \mid K_a\phi \mid D_G\phi \mid M_G\phi \\
(\langcdm)&\phi ::= p \mid \neg \phi \mid (\phi \ra \phi) \mid K_a\phi \mid C_G\phi \mid D_G\phi \mid M_G\phi.\\
\end{array}$$
In the above, ``$p$'' is a member of the set $\prop$ of propositional variables, ``$a$'' belongs to the set $\ag$ of agents, and ``$G$'' represents a nonempty subset of \ag, signifying a group. We also employ other boolean connectives, including conjunction ($\wedge$), disjunction ($\vee$), and equivalence ($\lra$). $E_G \phi$ is a shorthand for $\bigwedge_{a\in G} K_a\phi$. These are perceived as defined operators and are manipulated in the conventional manner.

In future instances, we may need to modify the parameter \ag, e.g., for the language \langm where the agents are drawn from a different set $\ag'$. In such cases, we will denote the set of agents explicitly as a parameter of the language. For example, we might write $\langm_{\ag}$ and $\langm_{\ag'}$ to distinguish between the different agent sets.

For any language \langl introduced above, an ``\emph{\langl-formula}'' or a ``\emph{formula of \langl}'' refers to a well-formed sentence of the language \langl. By a \emph{formula}, we generally mean a formula of any of the languages introduced above.

For any languages \langl and $\langl'$ introduced above, we say ``\langl is a \emph{sublanguage} of $\langl'$'' or ``$\langl'$ is a \emph{superlanguage} of \langl'', if every formula of \langl is also a formula of $\langl'$.
\qed
\end{definition}

Our study employs formulas such as $K_a \phi$ to depict ``agent $a$ knows $\phi$''. This is occasionally referred to as \emph{individual knowledge}. Similarly, formulas like $C_G\phi$, $D_G\phi$, $E_G\phi$ and $M_G\phi$ are used to convey that ``$\phi$ is \emph{common knowledge}, \emph{distributed knowledge}, \emph{everyone's knowledge} or \emph{mutual knowledge} of group $G$'', respectively.

Before delving into the formal semantics of these formulas, it is important to first establish the semantic models that will be used for the intended logics.

\subsection{Semantic models}\label{sec:models}

We introduce two types of graphs defined in line with the traditional definitions established in graph theory. The first type is referred to as \emph{(weighted) graphs}. However, in our context, the weights in these graphs are not arbitrary; instead, they represent the degrees of similarity or uncertainty among various epistemic objects such as specified data or areas of interest. These objects are represented by the nodes within the graphs. Building on this concept, we further apply specific constraints established in data mining to the graphs to derive a specialized subtype, which we term as \emph{similarity graphs}. The similarity graphs, while maintaining the foundational structure of graphs, are tailored to facilitate the interpretation of epistemic languages, a function that will be further elaborated upon in Section~\ref{sec:semantics}.

\begin{definition}[graphs]\label{def:graphs}
A \emph{(weighted) graph} is represented by the tuple $(W,\ab,E)$, where:
\begin{itemize}
\item $W$ is a nonempty set of states or nodes, referred to as the \emph{domain};
\item \ab is an arbitrary set of abstract epistemic abilities, which could be empty, finite or infinite, depending on the context;
\item $E : W \times W \to \wp(\ab)$, known as an \emph{edge function}, maps each pair of states to a set of epistemic abilities. This implies that the two states are indistinguishable for individuals possessing only these epistemic abilities.
\end{itemize}
A graph $(W,\ab,E)$ is referred to as a \emph{similarity graph} if it satisfies the following conditions for all $s,t \in W$:
\begin{itemize}
\item Positivity: If $E(s,t) = \ab$, then $s = t$;
\item Symmetry: $E(s,t) = E(t,s)$.
\qed
\end{itemize}
\end{definition}

The above definition warrants further elucidation. Firstly, our approach adopts a broad interpretation of epistemic abilities that may not necessarily be arranged in an order, although such an arrangement is plausible. Secondly, we perceive the edge function $E$ as a representation of the relation of similarity between states. In this context, similarities are deemed objective, signifying their constancy across diverse agents. Thirdly, the conditions of positivity and symmetry serve as generalized forms of common conditions employed to characterize similarity between data objects, as demonstrated in \cite{TSK2005}.%
\footnote{An implicit condition often assumed, the converse of positivity, posits $E(s,t)=A$ if $s=t$. This condition entails the reflexivity of graphs, depicted by the characterization axiom T (i.e., $K_a\phi \ra \phi$). In the realm of data mining, this condition implies that if two data objects are identical (i.e., they refer to the same data object), they would receive the maximum value from any similarity measure. However, given the typical lack of necessity for this condition, we have elected to omit it from this paper.}

\begin{example}[scenario for a graph]\label{ex1}
Consider a collaborative scenario involving three authors working on a research paper that is divided into four main sections: the introductory and motivational part (section 1), the introduction and definition of new logics (section 2), the axiomatization of these logics (section 3), and the calculation of the complexity of these logics (section 4). Due to differing perspectives on the paper's presentation and poor synchronization of their efforts, the collaboration yields five versions of the paper, denoted as $s_1, s_2, \dots, s_5$. Despite having varying texts, these variants may not inherently differ from each other -- they could have similar meaning with subtle or significant differences.

Let us consider the following circumstances:
\begin{itemize}
\item In section 1, the variants $s_1$, $s_2$, $s_3$ and $s_5$ are largely similar, while $s_4$ diverges.
\item For sections 2 and 3, $s_1$ and $s_3$ share common content, while $s_2$, $s_4$ and $s_5$ -- also similar to each other -- differ substantially from $s_1$ and $s_3$.
\item In section 4, variants $s_1$ and $s_2$ align closely, differing from $s_3$, $s_4$ and $s_5$, which, while similar to each other, deviate from the first two variants.
\qed
\end{itemize}
\end{example}

The situation in Example~\ref{ex1} can be formalized with the graph $G = (W,\ab,E)$, where $W = \{ s_1, s_2, s_3, s_4, s_5 \}$, $\ab = \{ 1, 2, 3, 4 \}$, and $E$ is such that:
$E(s_1,s_1) = E(s_2,s_2) = E(s_3,s_3) = E(s_4,s_4) = E(s_5,s_5) = \{ 1, 2, 3, 4 \}$,
$E(s_1,s_2) = E(s_3, s_5) = \{ 1, 4\}$,
$E(s_1,s_3) = E(s_2,s_5) = \{ 1, 2, 3\}$,
$E(s_1,s_4) = \emptyset$,
$E(s_1,s_5) = E(s_2, s_3) = \{ 1 \}$,
$E(s_2,s_4) = \{ 2, 3 \}$,
$E(s_3,s_4) = \{ 4 \}$,
$E(s_4,s_5) = \{ 2, 3, 4 \}$,
and is symmetric (i.e., for all $x,y \in W$, $E(x,y) = E(y,x)$). $G$ is in fact a similarity graph, and is illustrated in Figure~\ref{fig:sim-graph}.%

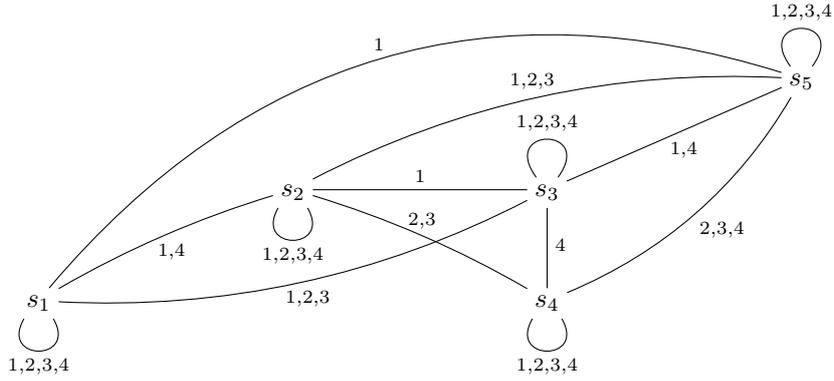
\begin{figure}
\centering
$\xymatrix@R=3em@C=8em{
&&&s_5
\ar@{-}@(ul,ur)^{1, 2, 3, 4}
\\
&s_2
\ar@{-}@(dl,dr)_{1, 2, 3, 4}
\ar@{-}[r]^{1}
\ar@{-}@/^.3pc/[dr]^{\hspace{-6pt}2, 3}
\ar@{-}@/^1.2pc/[urr]^{1, 2, 3}
& s_3
\ar@{-}@(ul,ur)^{1, 2, 3, 4}
\ar@{-}[d]^{4}
\ar@{-}[ru]_{1, 4}
\\
s_1
\ar@{-}@(dl,dr)_{1, 2, 3, 4}
\ar@{-}@/^.3pc/[ur]_{1, 4}
\ar@{-}@/_1.2pc/[rru]_{1, 2, 3}
\ar@{-}@/^4.3pc/[uurrr]^{1}
&
&s_4
\ar@{-}@(dl,dr)_{1, 2, 3, 4}
\ar@{-}@/_1pc/[uur]_{2,3,4}
}$
\caption{Illustration of a similarity graph for Example~\ref{ex1}. We do not draw a line between two nodes when the edge between them is with no label (or in other words, labeled by an empty set; see the case between $s_1$ and $s_4$).}\label{fig:sim-graph}
\end{figure}

Now, we incorporate a capability function and a valuation into each of the graphs, resulting in the \emph{models} that will be used to interpret the formal languages.

\begin{definition}[models]\label{def:models}
A \emph{model} is a quintuple $(W,\ab,E,C,\nu)$ such that:
\begin{itemize}
\item $(W,\ab,E)$ forms a graph,
\item $C: \ag \to \wp(\ab)$ is a capability function that assigns each agent a set of epistemic abilities,
\item $\nu: W \to \wp(\prop)$ is a valuation that assigns a set of propositional variables (representing those that are true) to every state.
\end{itemize}

A model $(W,\ab,E,C,\nu)$ is referred to as:
\begin{itemize}
\item a \emph{symmetric model}, if $E$ satisfies symmetry (see Definition~\ref{def:graphs});
\item a \emph{similarity model}, if $(W,\ab,E)$ forms a similarity graph (i.e., it complies with positivity and symmetry; see Definition~\ref{def:graphs}).
\qed
\end{itemize}
\end{definition}

The next example of a similarity model continues from Example~\ref{ex1}.

\begin{example}[scenario for a model]\label{ex2}
The three authors are identified as $a$, $b$ and $c$. Author $a$ excels at generating new ideas (section 1), introducing and defining new logics (section 2), and axiomatizing logics (section 3). Author $b$ is proficient in the topics of sections 2 and 3, as well as calculating complexity of logics (section 4). Author $c$ is adept only at the topic of section 4.

It is crucial for the authors to ascertain whether the paper accurately represents their intended content, but it is challenging for them to evaluate the sections beyond their expertise. Four propositions, $p_1$, $p_2$, $p_3$ and $p_4$, signify that the sections $1$ to $4$ accurately convey the planned content, respectively.

The introductory section (section 1) and the introduction of logics (section 2) of the variants $s_1$ are well-written, accurately presenting the planned content. Also commendably written are the sections 1 and 3 of $s_2$, the sections 1, 2 and 4 of $s_3$, the sections 3 and 4 of $s_4$, and the sections 1, 3 and 4 of $s_5$. The remaining sections contain misrepresentations.
\qed
\end{example}

Example~\ref{ex2} can be formalized in the (similarity) model $M = (W,\ab,E,C,\nu)$ where:
\begin{itemize}
\item $(W,\ab,E)$ forms the similarity graph for Example~\ref{ex1};
\item $C$ is such that $C(a) = \{1,2,3\}$, $C(b) = \{2,3,4\}$ and $C(c) = \{4\}$;
\item $\nu$ is such that $\nu(s_1) = \{p_1, p_2\}$, $\nu(s_2) = \{p_1,p_3\}$, $\nu(s_3) = \{p_1, p_2, p_4\}$, $\nu(s_4) = \{p_3,p_4\}$ and $\nu(s_5) = \{p_1,p_3,p_4\}$.
\end{itemize}
Figure~\ref{fig:sim-model} illustrates the similarity model $M$ introduced above.

\begin{figure}
\centering
\parbox{.75\textwidth}{
\centering
$\xymatrix@R=3em@C=6em{
&&
&*++o[F]{\frac{s_5}{p_1,p_3,p_4}}
\ar@{-}@(ul,ur)^{1,2,3,4}
\\
& *++o[F]{\frac{s_2}{p_1,p_3}}
\ar@{-}@(ul,ur)^{1,2,3,4}
\ar@{-}[r]^{1}
\ar@{-}@/^1pc/[dr]^{\hspace{-6pt}2,3}
\ar@{-}@/^1pc/[urr]^{1,2,3}
&*++o[F]{\frac{s_3}{p_1,p_2,p_4}} 
\ar@{-}@(ul,ur)^{1,2,3,4}
\ar@{-}[d]^{4}
\ar@{-}[ru]_{1,4}
\\
*++o[F]{\frac{s_1}{p_1,p_2}}
\ar@{-}@(dl,dr)_{1,2,3,4}
\ar@{-}@/^.3pc/[ur]_{1,4}
\ar@{-}@/_1.2pc/[rru]_{1,2,3}
\ar@{-}@/^5pc/[uurrr]^{1}
&
&*++o[F]{\frac{s_4}{p_3,p_4}}
\ar@{-}@(dl,dr)_{1,2,3,4}
\ar@{-}@/_1pc/[uur]_{2,3,4}
}$
}
\parbox{.23\textwidth}{
$\begin{array}{l}
C(a) = \{1,2,3\}\\
C(b) = \{2,3,4\}\\
C(c) = \{4\}\\
\end{array}$
}
\caption{Illustration of a similarity model for Example~\ref{ex2}.}\label{fig:sim-model}
\end{figure}
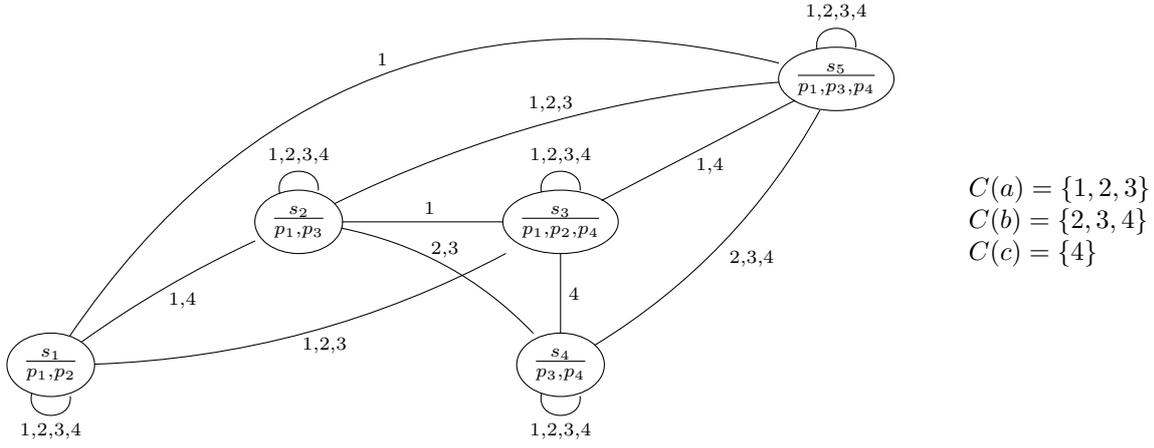

\subsection{Semantics}\label{sec:semantics}

Formulas are interpreted by the models introduced in the previous section, as made more precisely in the following definition.

\begin{definition}[satisfaction]\label{def:semantics}
Given a formula $\phi$, a model $M = (W,\ab,E,C,\nu)$ and a state $s \in W$, we say $\phi$ is \emph{true} or \emph{satisfied} at $s$ in $M$, denoted $M,s \models \phi$, if the following recursive conditions are met:
\[
\begin{array}{lll}
M,s \models p & \iff & p\in \nu(s)\\
M,s \models \neg\psi & \iff & \text{not } M,s \models \psi\\
M,s \models (\psi \ra \chi) & \iff & M,s \models \psi \text{ materially implies } M,s \models\chi\\
M,s\models K_a\psi & \iff & \text{for all $t\in W$, if $C(a)\subseteq E(s,t)$ then $M,t \models \psi$}\\
M,s\models C_G\psi & \iff & \text{for all $n \in \mbN^+$, $M,s \models E_G^n \psi$}\\
M,s\models D_G\psi & \iff & \text{for all $t \in W$, if $\bigcup_{a\in G}C(a) \subseteq E(s,t)$ then $M,t \models \psi$}\\
M,s\models M_G\psi & \iff & \text{for all $t \in W$, if $\bigcap_{a\in G}C(a) \subseteq E(s,t)$ then $M,t \models \psi$,}\\
\end{array}
\]
where $E_G^n\psi$ is defined recursively as $E_G E_G^{n-1}\psi$, and $E_G^1\psi$ is treated as $E_G \psi$, i.e., $\bigwedge_{a\in G} K_a\psi$.
\qed
\end{definition}

In the definition above, the interpretation of $K_a \psi$ includes a condition ``$C(a) \subseteq E(s,t)$'', which intuitively means that ``agent $a$, with his abilities, cannot distinguish between states $s$ and $t$". Thus, the formula $K_a \psi$ expresses that $\psi$ is true in all states $t$ that $a$ cannot differentiate from the current state $s$.

\emph{Common knowledge} ($C_G\psi$) follows the classical fixed-point interpretation as $E_G C_G \psi$, where $E_G \chi$ stands for the conventional notion of \emph{everyone's knowledge}, stating that ``everyone in group $G$ knows $\chi$'' (refer to \cite{HM1992} for more details). In other words, $C_G\psi$ implies that ``everyone in group $G$ knows that $\psi$ is true, and everyone in $G$ knows about this first-order knowledge, and also knows about the second-order knowledge, and so on".

The concept of \emph{Distributed knowledge} ($D_G\psi$) in this paper diverges from the traditional definitions found in literature. In classical terms, knowledge is distributed such that it is accumulated by pooling individual knowledge together (see \cite[Chapter 1]{HM1992}), though its classical formal definition gives rise to controversy \cite{Roelofsen2007,AW2017rdk}. We redefine distributed knowledge as being attainable by pooling together individual abilities. In practice, we swap the intersection of individual uncertainty relations with the union of individual epistemic abilities. Thus, $\psi$ is deemed distributed knowledge among group $G$ if and only if $\psi$ holds true in all states $t$ that, when utilizing all the epistemic abilities of agents in group $G$, cannot be differentiated from the present state.

An additional type of group knowledge, termed \emph{mutual knowledge} ($M_G\psi$), states that $\psi$ is mutual knowledge if and only if $\psi$ is true in all states $t$ that, using the mutual abilities of group $G$, cannot be differentiated from the current state. This concept is akin to everyone's knowledge%
\footnote{In the area of epistemic logic, $E_G \phi$ intuitively says, ``Everyone/everybody in group $G$ knows $\phi$.'' This is why the symbol $E$ is used in the formula, and this symbol has been made popular since the publications of the seminal textbooks on (dynamic) epistemic logic \cite{FHMV1995,MvdH1995,vDvdHK2008} . Yet we lack a name of what kind of knowledge $\phi$ is when $E_G\phi$ holds, a name that goes with ``common knowledge'' and ``distributed knowledge'' in a fancy way. This is the reason why we have named it ``everyone's knowledge''. In the literature, however, there are cases where ``mutual knowledge'' is used to refer to everyone's knowledge (see, e.g., \cite{VS2022}), a name reserved by us for a different concept of knowledge that is similar to everyone's knowledge.}%
, and we will examine its logical properties in greater detail later on.

\begin{example}
Using the similarity model for Example~\ref{ex2}, we have the following truths:
\begin{enumerate}
\item $M,s_2 \models K_a p_3$, meaning that when the variant $s_2$ is at hand, $a$ knows that its third section is well written.
\item $M,s_4 \models \neg K_b p_1 \wedge \neg K_b \neg p_1$. When $b$ sees the variant $s_4$, he does not know whether its first section is well written.
\item $M,s_3 \models K_c (K_a p_3 \vee K_a \neg p_3)$. When $c$ proofreads the paper $s_3$, she knows that $a$ knows whether its third section is well written.
\item $M,s_4 \models E_{\{a,b\}} (p_3 \wedge p_4)$. For the variant $s_4$, both $a$ and $b$ know that its third and forth sections are well written.
\item $M,s_5 \models (\neg C_{\{a,c\}} p_1 \wedge \neg C_{\{a,c\}} \neg p_1) \wedge (\neg C_{\{a,c\}} p_2 \wedge \neg C_{\{a,c\}} \neg p_2)$. It is not common knowledge for $a$ and $c$ whether either of the first two sections of variant $s_5$ is well written.
\item $M,s_4 \models D_{\{a,b\}} (\neg p_1 \wedge p_4)$. The first section of variant $s_4$ is not well written but the fourth section is -- this is distributed knowledge between $a$ and $b$.
\item $M,s_4 \models \neg M_{\{a,b\}} \neg p_1 \wedge \neg M_{\{a,b\}} p_4$. That the first section of variant $s_5$ is not well written is not mutual knowledge between $a$ and $b$, neither is the fourth section of $s_4$.
\end{enumerate}
\end{example}

\begin{definition}[satisfiability and validity]
A formula is referred to as:
\begin{itemize}
\item \emph{satisfiable}, if there exists a state in a model in which it is satisfied;
\item \emph{valid}, if it is satisfied in all states across all models;
\item \emph{s-satisfiable}, if there exists a state in a similarity model where it is satisfied;
\item \emph{s-valid}, if it is satisfied in all states across all similarity models.
\end{itemize}
A set of formulas is called \emph{satisfiable} (resp., \emph{s-satisfiable}), if there exists a state in a model (resp., \emph{similarity model}) in which all the formulas of the set is satisfied. A set of formulas is called \emph{valid} (resp., \emph{s-valid}), if all of its elements are valid (resp., s-valid).
\qed
\end{definition}

By Definition~\ref{def:models} it is clear that all similarity models are models. Therefore  it is intuitively evident that if a formula is s-satisfiable, it is also satisfiable, and if a formula is valid, it is likewise s-valid. When we refer to a ``logic'', we are talking about the set of valid or s-valid formulas derived from a specific language. For instance, the term \lcdm refers for the set of all \emph{valid} \langcdm-formulas, while \lscdm represents the set of all \emph{s-valid} \langcdm-formulas. The following definition elaborates on this convention.

\begin{definition}[notation for logics]
For any combination \text{\normalfont X} composed of elements of the set $\{\text{\normalfont C, D, M}\}$, like \text{\normalfont C}, \text{\normalfont CDM}, or even an empty string, we refer to its calligraphic transformation $\mathcal{X}$ as the combination of the calligraphic symbols of \text{\normalfont X}. This could manifest as $\mathcal{C}$, $\mathcal{CDM}$, or an empty string.

Given that \text{\normalfont X} and $\mathcal{X}$ can be any of the above-mentioned combinations, the notation \text{\normalfont ELX} is used to represent the set of all valid $\mathcal{ELX}$-formulas, and \text{\normalfont ELX$^s$} denotes the set of all s-valid $\mathcal{ELX}$-formulas.
\qed
\end{definition}

\begin{proposition}
The following hold for any formula $\phi$, any agent $a$ and any groups $G$ and $H$:
\begin{enumerate}
\item $K_a(\phi \ra \psi) \ra (K_a\phi \ra K_a\psi)$ is valid (hence s-valid);
\item $\phi \ra K_a \neg K_a \neg \phi$ is not valid, but is s-valid;
\item $C_{G}\phi \ra \bigwedge_{a\in G} K_a(\phi \wedge C_G \phi)$ is valid (hence s-valid);
\item $D_{\{a\}}\phi \lra K_a\phi$ is valid (hence s-valid);
\item $D_G \phi \ra D_H \phi$ (with $G \subseteq H$) is valid (hence s-valid);
\item $\phi \ra D_G \neg D_G \neg \phi$ is not valid, but is s-valid;
\item $M_{\{a\}}\phi \lra K_a\phi$ is valid (hence s-valid);
\item $M_G \phi \ra M_H \phi$ (with $H \subseteq G$) is valid (hence s-valid);
\item $\phi \ra M_G \neg M_G \neg \phi$ is not valid, but is s-valid.
\end{enumerate}
\end{proposition}
\begin{proof}
We only show the last clause for example. Consider the model $M=(W,\ab,E,C,\mu)$, where $W=\{s,t\}$, $\ab=\{1\}$, $E(s,t)=C(a)=\{1\}$, $E(t,s)=E(s,s)=E(t,t)=\emptyset$ and $\nu(s)=\nu(t)=\emptyset$. In this context, although $M,s \models \top$, we find that $M,s \not \models M_{\{a\}} \neg M_{\{a\}}\neg \top$. This indicates that $\phi \ra M_G \neg M_G \neg \phi$ is invalid. To establish that $\phi \ra M_G \neg M_G \neg \phi$ is s-valid, let us consider a state $s$ in any similarity model $M$. Assuming $M,s \models \phi$, then for any state $t$ such that $\bigcap_{a\in G}C(a)\subseteq E(s,t)$, we find that $\bigcap_{a \in G} C(a) \subseteq E(t,s)$. Therefore, $M,t \not \models M_G\neg\phi$. As a result, $M,s \models M_G \neg M_G \neg \phi$. This implies that $\phi \ra M_G \neg M_G \neg \phi$ is indeed s-valid.
\end{proof}

\subsection{Expressivity}
\label{sec:expressivity}

We utilize the traditional approach of determining the expressive power of a language, which involves comparing its relative expressive power with other languages.

Consider $\langl_1$ and $\langl_2$ as two languages whose formulas can be validated as true or false within the same type of structures, such as the proposed \lang and \langc. We state that $\langl_1$ is \emph{at most as expressive as} $\langl_2$, denoted as $\langl_1 \preceq \langl_2$, if for any $\langl_1$-formula $\phi$ there exists an $\langl_2$-formula $\psi$ such that $\phi$ and $\psi$ are \emph{equivalent}. Here, equivalence means that the two formulas are true or false in the exact same states of the same models, according to the semantics. We denote $\langl_1 \equiv \langl_2$ if both $\langl_1 \preceq \langl_2$ and $\langl_2 \preceq \langl_1$ are true. We claim that $\langl_1$ is less expressive than $\langl_2$ (strictly speaking), denoted as $\langl_1 \prec \langl_2$, if $\langl_1 \preceq \langl_2$ and $\langl_1 \not\equiv \langl_2$ hold true. We assert that $\langl_1$ and $\langl_2$ are \emph{incomparable} if neither $\langl_1 \preceq \langl_2$ nor $\langl_2 \preceq \langl_1$ holds true.

Clearly, all eight languages proposed can be compared to each other in terms of their expressive power, either with respect to the class of all models or the class of all similarity models. The classification of models, either as all models or all similarity models, does not influence the result, thereby enabling a uniform presentation (see Figure~\ref{fig:expressivty} for a summary of the results).

\begin{figure}[htbp]
\footnotesize
\subcaptionbox{when $|\ag|=1$}[.49\textwidth]{
$\xymatrix{
&*++o[F]{\langcdm}\ar@<1pt>[dl]\ar@<1pt>[d]
\\
*++o[F]{\langcd}\ar@<1pt>[ur]\ar@<1pt>[d]&
*++o[F]{\langcm}\ar@<1pt>[u]\ar@<1pt>[dl]|(.5)\hole&
*++o[F]{\langdm}\ar[ul]\ar@<1pt>[dl]\ar@<1pt>[d]
\\
*++o[F]{\langc}\ar@<1pt>[u]\ar@<1pt>[ur]|(.5)\hole&
*++o[F]{\langd}\ar[ul]\ar@<1pt>[ur]\ar@<1pt>[d]&
*++o[F]{\langm}\ar[ul]|(.5)\hole\ar@<1pt>[u]\ar@<1pt>[dl]
\\
&*++o[F]{\lang}\ar[ul]\ar@<1pt>[u]\ar@<1pt>[ur]
}$
}
\subcaptionbox{when $|\ag| \geq 2$}[.49\textwidth]{
$\xymatrix{
&*++o[F]{\langcdm}
\\
*++o[F]{\langcd}\ar[ur]&*++o[F]{\langcm}\ar[u]&*++o[F]{\langdm}\ar[ul]
\\
*++o[F]{\langc}\ar[u]\ar[ur]|(.5)\hole&*++o[F]{\langd}\ar[ul]\ar[ur]&*++o[F]{\langm}\ar[ul]|(.5)\hole\ar[u]
\\
&*++o[F]{\lang}\ar[ul]\ar[u]\ar[ur]
}$
}
\caption{This figure illustrates the relative expressive power of the languages. An arrow pointing from one language to another implies that the first language is at most as expressive as the second. The ``at most as expressive as'' relationship is presumed to be ``transitive'', meaning that a language is considered at most as expressive as another if a path of arrows exists leading from the first to the second. A lack of a path of arrows from one language to another indicates that the first language is \emph{not} at most as expressive as the second. This implies that either the two languages are incomparable or that the second language is less expressive than the first.\label{fig:expressivty}}
\end{figure}
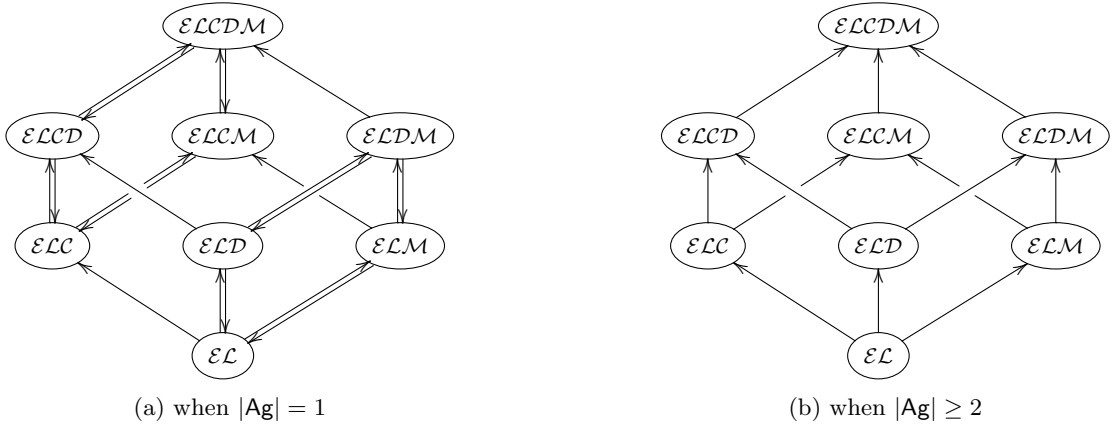

In Figure~\ref{fig:expressivty}, every language, with the exception of \langcdm, has an arrow pointing to its immediate superlanguages. This is clearly true, as by definition, every language is at most as expressive as its superlanguages. In the case when $|\ag|=1$, a reverse arrow also exists between languages that either both contain common knowledge or neither contain common knowledge.

\begin{lemma}\label{lem:exp1}
When $|\ag|= 1$ (i.e., when there is only one agent available in the language), 
\begin{enumerate}
\item $\lang\equiv\langd\equiv\langm\equiv\langdm$
\item $\langc\equiv\langcd\equiv\langcm\equiv\langcdm$
\item\label{it:exp-c} $\langdm \prec \langc$, and hence any language in the first clause are less expressive than any language in the second clause.
\end{enumerate}
\end{lemma}
\begin{proof}
1 \& 2. In the case when $|\ag|=1$ there is only one agent, and since $D_{\{a\}} \phi$ and $M_{\{a\}} \phi$ are equivalent to $K_a \phi$, the operators for distributed and mutual knowledge are redundant in this case. Hence the lemma.

3. We show that $\langc \not\preceq \langdm$, and so $\langdm \prec \langc$ since $\langdm\equiv\lang \preceq \langc$ by the first clause. Suppose towards a contradiction that there exists a formula $\phi$ of $\langdm$ equivalent to $C_{\{a\}}p$. Consider the set $\Phi=\{E_{\{a\}}^n p\mid n\in\mbN\}\cup\{\neg C_{\{a\}}p\}$. It is not hard to see that any finite subset of $\Phi$ is (s-)satisfiable, but not $\Phi$ itself. 
Let $k$ be the length of $\phi$ (refer to a modal logic textbook for its definition), and suppose $\{E_{\{a\}}^n p\mid n\in\mbN, n\leq k\}\cup\{\neg\phi\}$ is satisfied at a state $w$ in a (similarity) model $M=(W,\ab,E,C,\nu)$. For any $s,t \in W$, we say that $s$ reaches $t$ in one step if $C(a)\subseteq E(s,t)$. Consider the model $M_k=(W_k,\ab,E,C,\nu)$, where $W_k$ is set of states reachable from $w$ in at most $k$ steps. We can verify that $M_k, w \models \{E_{\{a\}}^n p\mid n\in\mbN\}\cup\{\neg\phi\}$, which implies that $\Phi$ is (s-)satisfiable, leading to a contradiction.
\end{proof}

We now proceed to elucidate the absence of arrows in the figure for the case when $|\ag| \geq 2$.

\begin{lemma}\label{lem:exp2}
When $|\ag| \geq 2$ (i.e., when there are at least two agents available in the language), 
\begin{enumerate}
\item For any superlanguage $\langl$ of \langc, and any sublanguage $\langl'$ of $\langdm$, it is not the case that $\langl \preceq \langl'$;
\item For any superlanguage $\langl$ of \langd, and any sublanguage $\langl'$ of $\langcm$, it is not the case that $\langl \preceq \langl'$;
\item For any superlanguage $\langl$ of \langm, and any sublanguage $\langl'$ of $\langcd$, it is not the case that $\langl \preceq \langl'$.
\end{enumerate}
\end{lemma}
\begin{proof}
1. The Proof of Lemma~\ref{lem:exp1}(\ref{it:exp-c}) can be used here to show that $\langc \not\preceq \langdm$ also when $|\ag| \geq 2$.

2. Consider models $M=(W,\ab,E,C,\nu)$ and $M'=(W',\ab,E',C,\nu')$, where $\ab=\{1,2,3\}$, $C(a)=\{1,2\}$, $C(b)=\{1,3\}$ (if there are more agents in the language, they are irrelevant here), and are illustrated below.
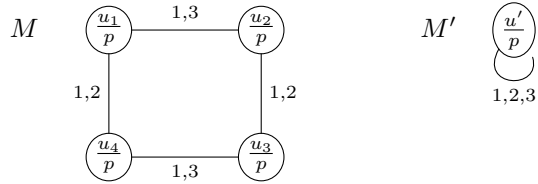
\begin{center}
$M$\quad$\xymatrix@R=3em@C=4em{
			*+o[F]{\frac{u_1}{p}} \ar@{-}[d]_{1,2} \ar@{-}[r]^{1,3} & *+o[F]{\frac{u_2}{p}} \ar@{-}[d]^{1,2} \\
			*+o[F]{\frac{u_4}{p}} \ar@{-}[r]_{1,3} & *+o[F]{\frac{u_3}{p}}
		}$
\qquad\qquad
$M'$\quad$\xymatrix@R=3em@C=4em{
			*+o[F]{\frac{u'}{p}} \ar@{-}@(dl,dr)_{1,2,3} 
		}$
\end{center}
Notice that both $M$ and $M'$ are similarity models.
We can show by induction that for any formula $\phi$ of \langcm, $M,u_1\models\phi$ iff $M',u'\models\phi$. On the other hand, $M,u_1\models D_{\{a,b\}}\bot$ but $M',u'\not\models D_{\{a,b\}}\bot$.
It means that no \langcm-formula can discern between $M,u_1$ and $M',u'$, while languages with distributed knowledge can. Thus the lemma holds.

3. Consider similarity models $M=(W,\ab,E,C,\nu)$ and $M'=(W',\ab,E',C,\nu')$, where $\ab=\{1,2,3\}$, $C(a)=\{1,2\}$, $C(b)=\{1,3\}$, and are illustrated below.
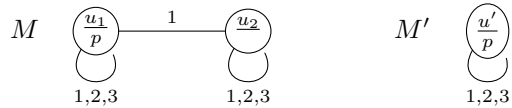
\begin{center}
$M$\quad$\xymatrix@R=3em@C=4em{
			*+o[F]{\frac{u_1}{p}} \ar@{-}@(dl,dr)_{1,2,3} \ar@{-}[r]^{1} & *+o[F]{\frac{u_2}{}} \ar@{-}@(dl,dr)_{1,2,3} 
		}$
\qquad\qquad
$M'$\quad$\xymatrix@R=3em@C=4em{
			*+o[F]{\frac{u'}{p}} \ar@{-}@(dl,dr)_{1,2,3} 
		}$
\end{center}
We can show by induction that for any formula $\phi$ of \langcd, $M,u_1\models\phi$ iff $M',u'\models\phi$. Meanwhile, we have $M,u_1\not\models M_{\{a,b\}}p$ and $M',u'\models M_{\{a,b\}}p$. It follows that no \langcd-formula can discern between $M,u_1$ and $M',u'$, while languages with mutual knowledge can. Thus the lemma holds.
\qed
\end{proof}

As per Figure~\ref{fig:expressivty}, Lemma~\ref{lem:exp2} suggests that there is not an arrow or a path of arrows leading from \langc (or any language having an arrow or a path of arrows originating from \langc) to \langdm (or any language with an arrow or a path of arrows pointing to \langdm). Similar relationships exist between \langd and \langcm, and between \langm and \langcd. Furthermore, in Figure~\ref{fig:expressivty}, if there is an arrow or a path of arrows leading from one language to another, and not the other way round, this signifies that the first language is less expressive than the second. If there is no arrow or path of arrows in either direction between two languages, they are deemed incomparable. These observations lead us directly to the following corollary.

\begin{corollary}
When $|\ag|\geq 2$,
\begin{enumerate}
\item $\lang \prec \langc$, $\langd \prec \langcd$, $\langm \prec \langcm$ and $\langdm \prec \langcdm$;
\item $\lang \prec \langd$, $\langc \prec \langcd$, $\langm \prec \langdm$ and $\langcm \prec \langcdm$;
\item $\lang \prec \langm$, $\langc \prec \langcm$, $\langd \prec \langdm$ and $\langcd \prec \langcdm$;
\item \langc, \langd and \langm are pairwise incomparable;
\item \langcd, \langcm and \langdm are pairwise incomparable;
\item \langc is incomparable with \langdm;
\item \langd is incomparable with \langcm;
\item \langm is incomparable with \langcd.
\end{enumerate}
\end{corollary}

\subsection{Semantic correspondence}
\label{sec:correspondence}

In this section, we introduce two truth-preserving translations: one from (similarity) models to classical \emph{relational models} (also known as \emph{Kripke models}) for modal logic, and the other in the reverse direction. These translations not only facilitate a comparison between the proposed logics and the classical epistemic modal logics, but also contribute to the completeness results of the proposed logics (see Section~\ref{sec:completeness1}). Familiarity with standard multi-agent epistemic logic interpreted over relational models is assumed.

In brief, a \emph{relational model} for multi-agent epistemic logic is a tuple $N = (W,R,V)$ where (i) $W$ represents the domain, (ii) $R : \ag \to \wp(W \times W)$ assigns each agent a binary relation $R(a)$ on $W$ (this need not to be an equivalence relation; any binary relation will suffice), and (iii) $V: \prop \to \wp(W)$ assigns each propositional variable a set of states. $N$ is referred to as a \emph{symmetric} relational model if every binary relation $R(a)$ of $N$ is symmetric. The satisfaction of a given formula $\phi$ at a state $s$ in $N$ (denoted $N,s \Vdash \phi$) is defined in the same way as in classical epistemic logic (see, for example, \cite{FHMV1995}). In particular, 
\begin{center}
\vspace{6pt}
\begin{tabular}{lll}
$N,s \Vdash K_a \phi$ &$\iff$& for all $t \in W$, if $(s,t) \in R(a)$ then $N,t \Vdash \phi$\\
$N,s \Vdash C_G \phi$ &$\iff$& for all $k \in \mbN^+$, $M,s \models E_G^k \psi$\\
$N,s \Vdash D_G \phi$ &$\iff$& for all $t \in W$, if $(s,t) \in \bigcap_{a \in G} R(a)$ then $N,t \Vdash \phi$\\
\end{tabular}
\vspace{6pt}
\end{center}
For the conventions of notations in the above (especially for the interpretation of common knowledge), we refer to Definition ~\ref{def:semantics}.

We note that the logics we propose can be reduced to their classical counterparts of multi-agent epistemic logic over relational models. However, as we do not have a parallel for mutual knowledge in classical epistemic logic, we limit ourselves to the logics without extending beyond the language \langcd. We present the following definition and lemma to introduce and explain the results.

\begin{definition}[standard translation]\label{def:trans-r}
A mapping $\cdot^\sigma$ from models to relational models is called the \emph{standard translation}, if for any given model $M = (W,\ab,E,C,\nu)$, $M^\sigma$ is the relational model $(W,R,V)$ with the same domain such that:
\begin{itemize}
\item $R$ is such that for every $a \in \ag$, $R(a) = \{ (s,t) \in W \mid C(a) \subseteq E(s,t)\}$;
\item $V$ is such that for every $p \in \prop$, $V(p) = \{s \in W \mid p \in \nu(s)\}$.
\qed
\end{itemize}
\end{definition}

\begin{lemma}[truth preservation of standard translation]\label{lem:trans-r}
For any \langcd-formula $\phi$, any model $M$ and any state $s$ of $M$, $M,s \models \phi$ iff $M^\sigma,s \Vdash \phi$.
\end{lemma}
\begin{proof}
By induction on $\phi$. Let $M=(W,\ab,E,C,\nu)$ be any model and $M^\sigma = (W, R, V)$ its standard translation. The atomic case is quite obvious by the definition of standard translation. Boolean cases follow easily from the definition of satisfaction. The case for common knowledge holds because the semantic definition is the same -- both defined inductively upon everyone's knowledge (and in turn upon individual knowledge). The cases that need elaboration are those for individual and distributed knowledge:
$$\begin{array}{@{}rcll@{}}
M,s \models K_a \psi & \iff & \text{for all $t \in W$, if $C(a) \subseteq E(s,t)$ then $M,t \models \psi$} &\text{(Def.~\ref{def:semantics})}\\
& \iff & \text{for all $t \in W$, if $(s,t) \in R(a)$ then $M,t \models \psi$} &\text{(Def.~\ref{def:trans-r})}\\
& \iff & \text{for all $t \in W$, if $(s,t) \in R(a)$ then $M^\sigma,t \Vdash \psi$} & \text{(IH)}\\
& \iff & M^\sigma,s \Vdash K_a\psi. &\\[1em]
M,s \models D_G \psi & \iff & \text{for all $t\in W$, if $\bigcup_{a\in G}C(a) \subseteq E(s,t)$, then $M,t \models \psi$} &\text{(Def.~\ref{def:semantics})}\\
& \iff & \text{for all $t\in W$, if $C(a)\subseteq E(s,t)$ for all $a\in G$, then $M,t \models \psi$} & \\
& \iff & \text{for all $t\in W$, if $(s,t) \in R(a)$ for all $a\in G$, then $M,t \models \psi$} &\text{(Def.~\ref{def:trans-r})}\\
& \iff & \text{for all $t\in W$, if $(s,t) \in \bigcap_{a \in G} R(a)$, then $M,t \models \psi$} & \\
& \iff & \text{for all $t\in W$, if $(s,t) \in \bigcap_{a \in G} R(a)$, then $M^\sigma,t \Vdash \psi$} & \text{(IH)}\\
& \iff & M^\sigma,s \Vdash D_G\psi. &\\
\end{array}
$$
The induction holds, and so we achieve the lemma.
\end{proof}

Perhaps less apparent is our ability to also reduce the classical multi-agent epistemic logic over relational models to the logics studied in this paper by applying a slight modification to the set of epistemic abilities.

\begin{definition}[reverse translation]\label{def:trans-w}
A mapping $\cdot^\rho$ from relational models to models is called the \emph{reverse translation}, if for any given relational model $N = (W,R,V)$, $N^\rho$ is the model $(W,\ag,E,C,\nu)$ with the same domain and:
\begin{itemize}
\item $E$ is such that for all $s,t \in W$, $E(s,t) = \{ a \in \ag \mid (s,t) \in R(a)\}$,
\item $C$ is such that for all $a \in \ag$, $C(a) = \{a\}$, and
\item $\nu$ is such that for all $s \in W$, $\nu(s) = \{ p \in \prop \mid s \in V(p)\}$.
\qed
\end{itemize}
\end{definition}
In the translated model $N^\rho$ of the aforementioned definition, the set of epistemic abilities is appointed as \ag. We use agents as labels of edges, which can intuitively be understood as an agent's inability to distinguish the ongoing state from current state when considering their epistemic abilities as a whole.
In the subsequent lemma, we demonstrate that this reverse translation preserves truth.

\begin{lemma}[truth preservation of reverse translation]\label{lem:trans-w}
For any formula $\phi$, any relational model $N$ and any state $s$ of $N$, $N,s \Vdash \phi$ iff $N^\rho,s \models \phi$.
\end{lemma}
\begin{proof}
Let $N = (W,R,V)$ and its reverse translation $N^\rho = (W,\ag,E,C,\nu)$. We show the lemma by induction on $\phi$. The cases for the atomic propositions, boolean connections and common knowledge are easy to verify. Here we only show the cases for the individual and distributed knowledge:
\[
\begin{array}{rcll}
N,s \Vdash K_a\psi & \iff & \text{for all $t \in W$, if $(s,t) \in R(a)$ then $N,t \Vdash \psi$}& \\
& \iff & \text{for all $t \in W$, if $a \in E(s,t)$ then $N,t \Vdash \psi$} & \\
& \iff &  \text{for all $t \in W$, if $C(a) \subseteq E(s,t)$ then $N,t \Vdash \psi$} & \\
& \iff & \text{for all $t \in W$, if $C(a) \subseteq E(s,t)$ then $N^\rho,t \models \psi$} &\\
& \iff & N^\rho,s \models K_a\psi &\\[1em]
N,s \Vdash D_G\psi & \iff & \text{for all $t \in W$, if $(s,t) \in \bigcap_{a \in G} R(a)$, then $N,t \Vdash \psi$}& \\
& \iff & \text{for all $t\in W$, if $(s,t) \in R(a)$ for all $a\in G$, then $N,t \Vdash \psi$} & \\
& \iff & \text{for all $t\in W$, if $C(a)\subseteq E(s,t)$ for all $a\in G$, then $N,t \Vdash \psi$} & \\
& \iff & \text{for all $t\in W$, if $\bigcup_{a \in G} C(a)\subseteq E(s,t)$, then $N,t \models \psi$} &\\
& \iff & \text{for all $t\in W$, if $\bigcup_{a \in G} C(a)\subseteq E(s,t)$, then $N^\rho,t \models \psi$} &\\
& \iff & N^\rho,s \models D_G\psi &\\
\end{array}
\]
This finishes the proof.
\end{proof}

Additionally, the following lemma asserts that the standard and reverse translations are closed under symmetry.

\begin{lemma}\label{lem:sem-trans}
The following results hold for the translations:
\begin{enumerate}
\item\label{it:sound} Given a symmetric model $M$, its standard translation $M^\sigma$ is a symmetric relational model;
\item\label{it:comp} Given a symmetric relational model $N$, its reverse translation $N^\rho$ is a symmetric model.
\end{enumerate}
\end{lemma}
\begin{proof}
(\ref{it:sound}) Let $M=(W,\ab,E,C,\nu)$ be a symmetric model,  and its standard translation $M^\sigma = (W, R, V)$. For any $a\in \ag$ and $s,t\in W$, we have:
\[
\begin{array}{llll}
	(s,t)\in R(a) & \iff & C(a)\subseteq E(s,t) & \text{(Def.~\ref{def:trans-r})}\\
	& \iff & C(a)\subseteq E(t,s) & \text{(by symmetry, see Def.~\ref{def:graphs})}\\
	& \iff & (t,s)\in R(a). & \text{(Def.~\ref{def:trans-r})}\\
\end{array}
\]
Thus $M^\sigma$ is a symmetric relational model.

(\ref{it:comp}) Let $N = (W,R,V)$ be a symmetric relational model, and its reverse translation $N^\rho = (W,\ag,E,C,\nu)$. For any $a \in \ag$ and $s,t\in W$, we have:
\[
\begin{array}{llll}
	a\in E(s,t) & \iff & (s,t)\in R(a) & \text{(Def.~\ref{def:trans-w})}\\
	& \iff & (t,s)\in R(a) & \text{(since $R(a)$ is symmetric)}\\
	& \iff & a\in E(t,s). & \text{(Def.~\ref{def:trans-w})}\\
\end{array}
\]
Hence $N^\rho$ is a symmetric model.
\end{proof}

Lastly, we introduce a lemma for a distinctive type of ``translation,'' transitioning from a symmetric model to a similarity model, which preserves the truth of formulas.

\begin{lemma}\label{lem:tr-sim-model}
Given a symmetric model $M$ and a state $s$ of $M$, there exists an \emph{equivalent} similarity model sharing the same domain with $M$; namely, a model $M'$ with the same domain as $M$, such that for all \langcdm-formulas $\phi$, $M,s \models \phi$ iff $M',s \models \phi$.
\end{lemma}
\begin{proof}
Let $M = (W, \ab,E,C,\nu)$ be a symmetric model, and define $M' = (W, \ab \cup \{b\},E,C,\nu)$ with $b$ a new epistemic ability (i.e., $b \notin \ab$). $M'$ is clearly a symmetric model sharing the same domain with $M$, since it shares the same edge function $E$ as well. Furthermore, $M'$ satisfies positivity, as there cannot be any $s,t\in W$ such that $E(s,t) = \ab \cup \{b\}$ and $s \neq t$. Therefore $M'$ is a similarity model. Moreover, we can show by induction that for any formula $\phi$, $M, s \models \phi$ iff $M', s \models \phi$. The proof is straightforward in all cases. In particular, in the cases for the modal operators $K_a$, $C_G$, $D_G$ and $M_G$, since $E$ and $C$ keep the same in both models (the new element $b$ does not show up in the range of $C$), so the equivalence of $\phi$ can be easily obtained by the definitions and the induction hypothesis. 
\end{proof}

\section{Axiomatization}
\label{sec:ax}

The axiomatic systems that we will demonstrate to be sound and complete axiomatizations for the respective logics are listed in Table~\ref{fig:ax}. Let us delve into the details of these axiomatic systems.

\begin{table}[htbp]
\caption{Overview of the logics and their corresponding axiomatic systems examined in this paper.\label{fig:ax}}\centering
\begin{tblr}{%
	colspec={*{4}{l}},
	row{odd}={gray!10},
	row{1}={gray!20},
    vline{2-Y} = {1}{0.3pt,gray!50},
    vline{2-Y} = {2-Z}{0.3pt,gray!30},
    hline{1,2,Z} = {0.1pt,azure5},
}
Logic&Axiomatic system
&
Logic&Axiomatic system
\\
\l&\K&
\ls&\KB
\\
\lc&\KC&
\lsc&\KBC
\\
\ld&\KD&
\lsd&\KBD
\\
\lm&\KM&
\lsm&\KBM
\\
\lcd&\KCD&
\lscd&\KBCD
\\
\lcm&\KCM&
\lscm&\KBCM
\\
\ldm&\KDM&
\lsdm&\KBDM
\\
\lcdm&\KCDM&
\lscdm&\KBCDM
\\
\end{tblr}
\end{table}

\subsection{Detailed Description of Sixteen Axiomatic Systems}

The \K is a widely recognized axiomatic system for modal logic (here it refers to the multi-agent version with each $K_a$ functioning as a box operator). For simplicity, the axiom schemes are referred to as axioms in this context. System \K include PC and K, along with the rules MP and N  (see Figure~\ref{fig:K}). For a comprehensive understanding of these axiomatizations for modal logic, please refer to, say, \cite{BdRV2001}.

The axiom system \KB is derived by augmenting the system \K with an additional axiom B (see Figure~\ref{fig:K}). In this context, we represent \KB as \K $\oplus$ B, where the symbol $\oplus$ acts like a union operation for the set of axioms and/or rules.

\begin{figure}[htbp]
\centering
\begin{framed}
\begin{tblr}[t]{ll}
\SetCell[c=2]{c}Axioms of \K or \KB\\
\hline
(PC)&all instances of propositional tautologies\\
(K)&$K_a(\phi \ra \psi) \ra (K_a\phi \ra K_a\psi)$\\
\hline[dashed]
(B)&$\phi \ra K_a \neg K_a\neg \phi$ \hfill (\KB only)\\
\end{tblr}
\qquad
\begin{tblr}[t]{ll}
\SetCell[c=2]{c}Rules of \K and \KB\\
\hline
(MP)&from $\phi$ and $\phi \ra \psi$ infer $\phi$\\
(N)&from $\phi$ infer $K_a\phi$\\
\end{tblr}
\end{framed}
\caption{Axiomatic systems \K and \KB. System \KB includes all the axioms and rules listed, while system \K includes all except the axiom B.\label{fig:K}}
\end{figure}

Common knowledge is characterized by two inductive principles (Figure~\ref{fig:sysC}), represented by an axiom and a rule, which can be found in \cite{FHMV1995}. Here, the system \KC is represented as $\K \oplus \sysC$, and \KBC as $\KB \oplus \sysC$.

\begin{figure}[htbp]
\centering
\small
\begin{framed}
\begin{tblr}[t]{ll}
\SetCell[c=2]{c}System \sysC\\
\hline
(C1)&$C_{G}\phi \ra \bigwedge_{a\in G} K_a(\phi \wedge C_G \phi)$\\
(C2)&from $\phi \ra \bigwedge_{a\in G} K_a (\phi \wedge \psi)$ infer $\phi \ra C_G\psi$\\
\end{tblr}
\end{framed}
\caption{The system characterizing common knowledge.\label{fig:sysC}}
\end{figure}

Distributed knowledge is characterized by a set of additional axioms, the number and type of which depend on the base system. If the base system is \K, the characterization axioms for distributed knowledge form the system \sysD (as shown in Figure~\ref{fig:sysD}). The resulting system is denoted as $\KD = \K \oplus \sysD$, which we aim to prove as a sound and complete axiom system for the logic \ld. If the base system is \KB, the characterization axioms for distributed knowledge form the system \sysBD (depicted in Figure~\ref{fig:sysD}). The resulting system is then denoted as $\KBD = \KB \oplus \sysBD$, which we aim to validate as a sound and complete system for the logic \lsd.

\begin{figure}[htbp]
\centering
\small
\begin{framed}
\dashbox{
\begin{tblr}[t]{ll}
\SetCell[c=2]{c}System \sysD\\
\hline
(K$_D$)&$D_G(\phi \ra \psi) \ra (D_G\phi \ra D_G\psi)$\\
(D1)&$D_{\{a\}}\phi\lra K_a\phi$\\
(D2)&$D_G\phi\ra D_H\phi$ with $G\subseteq H$\\
\end{tblr}
}
\qquad
\dashbox{
\begin{tblr}[t]{ll}
\SetCell[c=2]{c}System \sysBD\\
\hline
&all the axioms of \sysD\\
(BD)&$\phi\ra D_G\neg D_G\neg\phi$\\
\end{tblr}
}
\end{framed}
\caption{Characterization axioms for distributed knowledge. Depending on whether the base system is \K or \KB, we have the sets \sysD and \sysBD respectively.\label{fig:sysD}}
\end{figure}

Mutual knowledge is characterized by the axiomatic systems \sysM and \sysBM (as illustrated in Figure~\ref{fig:sysM}). While these systems might appear similar to \sysD and \sysBD respectively, there is a critical distinction between the axioms M2 and D2 -- the groups $G$ and $H$ interchange their positions. The validity of these axioms can confirmed straightforwardly, and it is interesting to note how the varied sequence of $G$ and $H$ aligns perfectly with the union/intersection of epistemic abilities as seen in the semantics. Correspondingly, $\KM$ is represented as $\K \oplus \sysM$, and $\KBM$ as $\KB \oplus \sysBM$.

\begin{figure}[htbp]
\centering
\small
\begin{framed}
\dashbox{
\begin{tblr}[t]{ll}
\SetCell[c=2]{c}System \sysM\\
\hline
(K$_M$)&$M_G(\phi \ra \psi) \ra (M_G\phi \ra M_G\psi)$\\
(M1)&$M_{\{a\}}\phi\lra K_a\phi$\\
(M2)&$M_G\phi\ra M_H\phi$ with $H\subseteq G$\\
\end{tblr}
}
\qquad
\dashbox{
\begin{tblr}[t]{ll}
\SetCell[c=2]{c}System \sysBM\\
\hline
&all the axioms of \sysM\\
(BM)&$\phi \ra M_G \neg M_G\neg\phi$\\
\end{tblr}
}
\end{framed}
\caption{Characterization axioms for mutual knowledge. Depending on whether the base system is \K or \KB, we have the sets \sysM and \sysBM. A key difference between M2 and D2 from \sysD is the exchange of positions between $G$ and $H$.\label{fig:sysM}}
\end{figure}

Moving towards more complex axiomatic systems, they are constructed in a similar manner. For any given string $\chi$ comprising elements from the set $\mathbf{\{C,D,M\}}$:
\begin{itemize}
\item The axiomatic system $\mathbf{K(\chi)}$ consists of all axioms and rules of \K, along with those of the systems denoted by each character in string $\chi$;
\item The axiomatic system $\mathbf{KB(\chi)}$ integrates all axioms and rules of \KB and the systems $\mathbf{B(X)}$ for each $\mathbf{X}$ present in $\chi$ (the system $\mathbf{B(C)}$ is considered to be $\mathbf{C}$).
\end{itemize}
To illustrate, when $\chi$ is the string ``$\mathbf{CM}$'', $\KCM$ is represented $\K \oplus \sysC \oplus \sysM$, and $\KBCM$ as $\KB \oplus \sysC \oplus \sysBM$. For two extreme cases, firstly, when $\chi$ is an empty string, $\mathbf{K(\chi)}$ simply equates to \K and $\mathbf{KB(\chi)}$ equals $\KB$. Secondly, when $\chi$ is the string ``$\mathbf{CDM}$'', then $\KCDM$ equates to $\K \oplus \sysC \oplus \sysD \oplus \sysM$ and $\KBCDM$ equals $\KB \oplus \sysC \oplus \sysBD \oplus \sysBM$.

We now turn our attention to validating these axiomatic systems as appropriate axiomatizations for the corresponding logics, as outlined in Figure~\ref{fig:ax}. Generally, these are characterized by soundness and completeness results. Soundness signifies that all the theorems of an axiomatic system are valid sentences of the corresponding logic. This can be simplified to the task of verifying that all the axioms of the system are valid, and that all the rules preserve this validity. The soundness of the proposed axiomatic systems can be confirmed with relative ease. Though we omit the proof, we state it as the following theorem. We will follow this up with the completeness results in the subsequent section.

\begin{theorem}[soundness]
Every axiomatic system introduced in this section is sound for its corresponding logic, as listed in Figure~\ref{fig:ax}.
\qed
\end{theorem}

\subsection{Completeness}\label{sec:completness}

In this section, we aim to demonstrate the completeness of all sixteen axiomatic systems that were introduced earlier. When referring to a given axiomatic system  $\mathbf{\Lambda}$ for a logic $\lambda$, based on a defined language $L$, the term ``strong completeness'' of $\mathbf{\Lambda}$ (for the logic $\lambda$) is equivalent to: for any set $\Phi$ of $L$-formulas and any $L$-formula $\phi$, if $\phi$ is a logical consequence of $\Phi$ in $\lambda$ (denoted as $\Phi \models_\lambda \phi$), then $\phi$ is derivable in $\mathbf{\Lambda}$ from $\Phi$ (expressed as $\Phi \vdash_{\mathbf{\Lambda}} \phi$). This notion is synonymous with the idea that for any $\mathbf{\Lambda}$-consistent set $\Psi$ of $L$-formulas, $\Psi$ must be (s-)satisfiable (according to $\lambda$).

On the other hand, ``weak completeness'' constrains $\Phi$ to an empty set (or $\Psi$ to a singleton set), which is a specific instance of strong completeness. For more detailed information on this topic, we recommend referring to a textbook on modal logic, such as \cite{BdRV2001}.

It is a widely accepted fact in classical epistemic logic that the inclusion of common knowledge can cause a logic to lose its \emph{compactness}. This leads to the situation where its axiomatic system is not strongly complete, but only weakly complete (see, e.g., \cite{BdRV2001,vDvdHK2008}). This is also the case in our context. As a consequence, we will demonstrate that the eight systems that do not include common knowledge are strongly complete axiomatizations for their corresponding logics, while the other eight systems that do incorporate common knowledge are only weakly complete.

The structure of this section is predicated on the various proof techniques we employ. We start with a method that reduces the satisfiability from classical epistemic logics to the logics we have proposed (Section~\ref{sec:completeness1}). However, this technique is only applicable to a limited number of logics. For a more direct proof, we adapt the canonical model method for systems devoid of group knowledge, specifically \K and \KB. This adaptation can be achieved relatively simply (Section~\ref{sec:completeness2}). When dealing with systems that incorporate either distributed or mutual knowledge, but not both (i.e.,  \KD, \KM, \KBD and \KBM), we utilize a path-based canonical model method (Section~\ref{sec:completeness3}). For systems that include both distributed and mutual knowledge, namely \KDM and \KBDM, a slightly more nuanced approach is required. Despite this, we continue to apply the path-based canonical model method (Section~\ref{sec:completeness4}). Lastly, for the remaining eight systems incorporating common knowledge, we merge the finitary method (which involves constructing a closure) with the methods mentioned above (Section~\ref{sec:completeness5}).

\subsubsection{Proof by translation of satisfiability}
\label{sec:completeness1}

\begin{theorem}[completeness, part 1]\label{thm:completeness1}
The following hold:
\begin{enumerate}
\item \K is strongly complete for \l, which means that for any \lang-formula $\phi$ and any set $\Phi$ of \lang-formulas, if $\Phi$ semantically entails $\phi$ in \l (denoted as $\Phi \models_\l \phi$), then there is a proof in \K that $\Phi$ syntactically entails $\phi$ (denoted as $\Phi \vdash_{\K} \phi$).
\item \KC is weakly complete for \lc. In other words, for any \langc-formula $\phi$, if $\phi$ is valid in \lc (represented as $\models_\lc \phi$), then $\phi$ is a theorem of \KC (denoted as $\vdash_{\KC} \phi$).
\item \KD is strongly complete for \ld. This means that for any \langd-formula $\phi$ and any set $\Phi$ of \langd-formulas, if $\Phi \models_\ld \phi$, then $\Phi \vdash_{\KD} \phi$.
\item \KCD is weakly complete for \lcd, which means that for any \langcd-formula $\phi$, if $\models_\lcd \phi$ then $\vdash_{\KCD} \phi$.
\item \KB is strongly complete with respect to the class of all symmetric models.
\end{enumerate}
\end{theorem}
\begin{proof}
This proof relies heavily on understanding the concepts of semantic and syntactic entailment, and the relationship between them.
Suppose $\Phi \nvdash_{\K} \phi$. By the completeness of \K with respect to all relational models, a relational model $N$ and a state $s$ of $N$ exist such that $N,s$ satisfies all formulas in $\Phi$ ($N,s \Vdash \psi$ for all $\psi \in \Phi$) and $N,s$ satisfies the negation of $\phi$ ($N,s \Vdash \neg\phi$). By Lemma~\ref{lem:trans-w}, the reversely translated model $N^\rho$ is such that  $N^\rho,s \models \psi$ (for all $\psi \in \Phi$) and $N^\rho,s\models \neg\phi$. It follows that $\Phi \not\models_\l \phi$.
The other parts of the theorem can be demonstrated in a similar manner, depending on the completeness of those axiomatic systems in classical epistemic or modal logic (see \cite[Ch.~3]{FHMV1995} for \KC and \KD and \cite{WA2020} for \KCD). The case for \KB is part of the proof of Theorem~\ref{thm:completeness2}.
\end{proof}

The completeness with respect to similarity models is slightly more complex. The axiomatization is known to be complete for symmetric relational models, and can be translated back to a symmetric model in a way that preserves truth, but we need to construct a similarity model that satisfies an additional condition on positivity (see Definition~\ref{def:models}). Fortunately, it is not too difficult to construct a similarity model from a symmetric model in a way that preserves truth.

\begin{theorem}[completeness, part 2]\label{thm:completeness2}
\KB is strongly complete for \ls; that is, for any \lang-formula $\phi$ and any set $\Phi$ of \lang-formulas, if $\Phi \models_\ls \phi$, then $\Phi \vdash_{\KB} \phi$.
\end{theorem}
\begin{proof}
Consider an \lang-formula $\phi$ and a set $\Phi$ of \lang-formulas. Suppose $\Phi \nvdash_{\KB} \phi$. Given the completeness of \KB over symmetric relational models (a well-known fact in modal logic), there exists a symmetric relational model $N$ and a state $s$ of it such that $N,s \Vdash \psi$ (for all $\psi \in \Phi$) and $N,s \Vdash \neg\phi$. By Lemmas~\ref{lem:sem-trans}(\ref{it:comp}) and \ref{lem:trans-w}, the reverse translation of $N$, namely $N^\rho$, is a symmetric model that also satisfies all formulas in $\Phi$ ($N^\rho, s \models \psi$ for all $\psi \in \Phi$) and the negation of $\phi$ ($N^\rho, s \models \neg\phi$). Lemma~\ref{lem:tr-sim-model} assert that there exists a similarity model $M$, which shares the same domain with $N^\rho$, such that $M, s \models \psi$ (for all $\psi \in \Phi$) and $M, s \models \neg\phi$. Therefore, it can be concluded that $\Phi \not\models_\ls \phi$.
\end{proof}

It is possible to use the same method to achieve complete results for the other systems \KBC, \KBD and \KBCD. However, as far as we know, the completeness of these systems in relational semantics, while expected, has never been explicitly established. Therefore, we do not state the results here immediately, but rather present them as corollaries of the completeness proofs given next.

\subsubsection{Proof by the canonical model method}
\label{sec:completeness2}

In the previous section, we demonstrated that \K and \KB are complete axiomatizations for their corresponding logics using the method of translation. This method is efficient and relies on the completeness results for their counterparts in classical logics interpreted via relational semantics. However, the translation method cannot be employed for logics where such a result does not exist for their classical counterparts (for instance, \lsc, \lsd and \lscd), or when their classical counterparts have not been introduced or studied (for example, logics with mutual knowledge). In this section, we will provide direct proofs of the completeness of \K and \KB using the canonical model method and extend this to a completeness proof for other logics in later sections.

\paragraph{\underline{Completeness of \K}}

We first introduce the \emph{canonical model for EL}. The model for \ls and other logics can be adapted from this with minor modifications. Let us recall that \ag is the set of all agents.

\begin{definition}[canonical model for \l]\label{def:cm-el}
The canonical model for \l is denoted as a tuple $\CM = (\CW, \CA, \CE, \CC, \CV )$ where:
\begin{itemize}
\item $\CW$ is the collection of all \emph{canonical states}, i.e., maximal \K-consistent sets of \lang-formulas;
\item $\CA=\wp(\ag)$, the power set of all agents;
\item $\CE : \CW \times \CW \to \wp(\CA)$ is defined such that for any $\Phi,\Psi \in \CW$, $E (\Phi,\Psi) = \bigcup_{a\in \ag} E_a(\Phi,\Psi)$, where $$E_a(\Phi,\Psi) = \left\{\begin{array}{ll}
	\CC(a), & \text{if $\{\chi\mid K_a\chi\in\Phi\}\subseteq\Psi$},\\
	\emptyset, & \text{otherwise;}
\end{array}\right.$$
\item $\CC : \ag \to \wp(\CA)$ is defined such that for any agent $a$, $\CC(a)=\{G \subseteq \ag \mid a\in G \}$;
\item $\CV: \CW \to \wp(\prop)$ is defined such that for any $\Phi \in \CW$, $\CV (\Phi) = \{ p \in \prop \mid p \in \Phi \}$.
\qed
\end{itemize}
\end{definition}

It can be easily verified that the canonical model for \l indeed qualifies as a model. While \CW and \CV are defined in a similar manner as in the canonical model for a classical modal logic, the other three components need slight elaboration. In the canonical model, any group $G$ is defined to be a set of epistemic abilities (as per the definition of \CA). It behaves as the set of epistemic abilities that are common among the members of $G$. The ``canonical'' epistemic abilities of agent $a$, denoted as $\CC(a)$, represent the collection of all potential common abilities of groups that include $a$. $\CE_a$ is a \emph{canonical relation} that associates two canonical states in a standard way. The condition ``$\{\chi\mid K_a\chi\in\Phi\}\subseteq\Psi$'' used in the standard canonical model creates a connection from $\Phi$ to $\Psi$. In our canonical model, disconnectivity is labeled by the epistemic abilities required to achieve it. Therefore, we label the link from $\Phi$ to $\Psi$ by $\CC(a)$ when they are connected by the standard condition, and label the link by the empty set signifying their disconnection, if the standard condition is not satisfied.

\begin{lemma}[Truth Lemma]\label{lem:truth-K}
Let $\CM=(\CW,\CA,\CE,\CC,\CV)$ represent the canonical model for \l. For any $\Gamma \in \CW$ and any \lang-formula $\phi$, we have $\phi \in \Gamma$ if and only if $\CM,\Gamma \models_{\l} \phi$.
\end{lemma}
\begin{proof}
We will prove this lemma through induction on $\phi$. The base case and boolean cases can be easily established by the definition of $\nu$ and the induction hypothesis. However, the case where $\phi$ is $K_a\psi$ requires more careful consideration.

Assuming $K_a\psi\in\Gamma$, but $\CM,\Gamma\not\models_\l K_a\psi$, there would exist a $\Delta\in \CW$ such that $\CC(a)\subseteq \CE (\Gamma,\Delta)$ and $\CM,\Delta\not\models_\l \psi$. Consequently, $\{\chi\mid K_a\chi\in\Gamma\}\subseteq\Delta$ (otherwise $\{a\} \notin \CE(\Gamma,\Delta)$, contradicting $\{a\} \in \CC(a)$). Thus, $\psi\in\Delta$. It follows from the induction hypothesis that $\CM,\Delta \models_\l \psi$, which results in a contradiction.

Assuming $K_a\psi\notin\Gamma$, but $\CM,\Gamma\models_\l K_a\psi$, then for any $\Delta\in \CW$, $\CC(a)\subseteq \CE (\Gamma,\Delta)$ implies $\CM,\Delta\models_\l \psi$. Observe that $\{\neg\psi\}\cup\{\chi\mid K_a\chi\in\Gamma\}$ is \K-consistent. If not, then there would exists a finite set of formulas $\Delta_0 = \{\chi_1,\cdots,\chi_i\}$ with $1<i\in\mbN$, such that $\vdash_\K (\neg\psi \wedge \bigwedge\Delta_0) \ra \bot$, then $\vdash_\K (\bigwedge\Delta_0) \ra\psi$, which would lead to $\vdash_\K K_a( (\bigwedge\Delta_0)\ra\psi)$. This suggests $\vdash_\K (\bigwedge_{\chi\in\Delta_0}K_a\chi)\ra K_a\psi$ and subsequently, $K_a\psi\in\Gamma$ since it is closed under deduction, contradicting with $K_a\psi\notin\Gamma$. We can extend $\{\neg\psi\}\cup\{\chi\mid K_a\chi\in\Gamma\}$ to a maximal \K-consistent set $\Delta^+$ of formulas. It follows that $\CC(a)\subseteq \CE(\Gamma,\Delta^+)$. By the induction hypothesis, we have $\CM,\Delta^+\models_\l \neg\psi$, leading to a contradiction as well.
\end{proof}

\begin{theorem}[completeness of \K, with a direct proof]
For any \lang-formula $\phi$ and any set $\Phi$ of \lang-formulas, if $\Phi \models_\l \phi$, then $\Phi \vdash_\K \phi$.
\end{theorem}
\begin{proof}
Assume that $\Phi \not\vdash_\K \phi$, then $\Phi\cup\{\neg\phi\}$ is consistent. Extend $\Phi\cup\{\neg\phi\}$ to a maximal \K-consistent set $\Delta^+$ of formulas. Let $\CM=(\CW,\CA,\CE,\CC,\CV)$ be the canonical model for \l, we find that $\CM,\Delta^+\models \chi$ for any formula $\chi\in \Phi\cup\{\neg\phi\}$. It follows that $\Phi \not\models_\l \phi$.
\end{proof}

\paragraph{\underline{Completeness of \KB}}

\begin{definition}[canonical model for \ls]\label{def:cm-els}
The canonical model for \ls is an adaptation of the canonical model for \l (as per Definition~\ref{def:cm-el}) with some modifications:
\begin{itemize}
\item The set \CW is now the set of all maximal \KB-consistent sets of \lang-formulas;
\item The definition of the ``canonical'' edge function \CE is altered such that $\CE_a$ is defined as follows:
$$E_a(\Phi,\Psi) = \left\{\begin{array}{ll}
	\CC(a), & \text{if $\{\chi\mid K_a\chi\in\Phi\}\subseteq\Psi$ and $\{\chi\mid K_a\chi\in\Psi\}\subseteq\Phi$},\\
	\emptyset, & \text{otherwise.}
\end{array}\right.$$
\end{itemize}
In this case, $E_a$ is a commutative function which is necessary for the completeness of \KB.
\qed
\end{definition}

\begin{lemma}[canonicity]
The canonical model for \ls is a similarity model.
\end{lemma}
\begin{proof}
Let $\CM=(\CW,\CA,\CE,\CC,\CV)$ be the canonical model for \ls. It is straightforward to verify that \CM is a model. Furthermore, notice that $\emptyset \notin \CC(a)$ for any agent $a$, so $\CE(s,t)\neq\CA$ for any $s,t \in \CW$, ensuring positivity. The symmetry of the model is evident as $\CE(s,t)=\CE(t,s)$ for any $s,t \in \CW$. Therefore, $\CM$ is a similarity model.
\end{proof}

The Truth Lemma for \KB parallels that of Lemma~\ref{lem:truth-K}:

\begin{lemma}[Truth Lemma]
\label{lem:truth-KB}
Let $\CM = (\CW,\CA,\CE,\CC,\CV)$ be the canonical model for \ls. For any $\Gamma \in \CW$ and any \lang-formula $\phi$, we have $\phi \in \Gamma$ iff $\CM,\Gamma \models_{\ls} \phi$.
\end{lemma}
\begin{proof}
We will only demonstrate the case when $\phi$ is of the form $K_a \psi$ here. The direction from $K_a\psi\in \Gamma$ to $\CM,\Gamma\models_\ls K_a\psi$ can be shown in a manner similar to that in Lemma \ref{lem:truth-K}.

For the opposite direction, suppose $K_a\psi\notin \Gamma$, but $\CM,\Gamma\models_\ls K_a\psi$, then for any $\Delta\in \CW$, $\CC(a)\subseteq \CE (\Gamma,\Delta)$ implies $\CM,\Delta\models_\ls \psi$.
First, we assert that $\{\neg\psi\}\cup\{\chi\mid K_a\chi\in \Gamma\}\cup\{\neg K_a\neg\chi\mid \chi\in \Gamma\}$ is \KB consistent. If not, note that for any $\eta\in\{\chi\mid K_a\chi\in \Gamma\}$, we have $\neg K_a \neg K_a\eta\in \{\neg K_a\neg\chi\mid \chi\in \Gamma\}$. As $\vdash_\KB \neg K_a \neg K_a \eta \ra \eta$, it follows that $\{\neg\psi\}\cup\{\neg K_a\neg\chi\mid \chi\in \Gamma\}$ is not \KB consistent. Therefore, we have $\vdash_\KB \big(\bigwedge_{\chi\in\Gamma_0}\neg K_a\neg \chi \big) \ra \psi$ for some finite subset $\Gamma_0$ of $\Gamma$. This leads to $\vdash_\KB K_a\big((\bigwedge_{\chi\in\Gamma_0}\neg K_a\neg \chi) \ra \psi\big)$, and hence $\vdash_\KB \bigwedge_{\chi\in\Gamma_0}K_a\neg K_a\neg \chi \ra K_a\psi$. Since we have $\vdash_\KB \chi \ra K_a\neg K_a\neg \chi$ for any $\chi\in\Gamma_0$, it follows that we have $\vdash_\KB \big( \bigwedge_{\chi\in\Gamma_0}\chi \big) \ra K_a\psi$.This deduction implies that $K_a\psi \in \Gamma$, which contradicts our previous assumption.
Now, let us extend the set $\{\neg\psi\}\cup\{\chi\mid K_a\chi\in \Gamma\}\cup\{\neg K_a\neg\chi\mid \chi\in \Gamma\}$ to some maximal \KB-consistent set $\Delta^+$ of \lang-formulas. Notice that $K_a\chi\in \Gamma$ implies $\chi\in\Delta^+$ for any $\chi$. Furthermore, if we suppose $\chi\notin \Gamma$, then $\neg\chi\in \Gamma$, which leads to $\neg K_a\neg\neg\chi\in\Delta^+$, implying $\neg K_a\chi\in\Delta^+$. Therefore, $K_a\chi\in \Delta^+$ implies $\chi\in \Gamma$ for any $\chi$.
Given these stipulations, we find that $\CC(a)\subseteq\CE(\Gamma,\Delta)$. However, by using the induction hypothesis, we see that $M,\Delta\not\models_\lsm \psi$. As a result, $M,\Gamma\not\models_\lsm K_a\psi$. This conclusion contradicts our previous assumptions, confirming this direction of the lemma.
\end{proof}

With the Truth Lemma, we can state the following theorem:

\begin{theorem}[completeness of \KB, with a direct proof]
For any \lang-formula $\phi$ and any set $\Phi$ of \lang-formulas, if $\Phi \models_\ls \phi$, then $\Phi \vdash_\KB \phi$.
\end{theorem}
\begin{proof}
To prove this, suppose the contrary: $\Phi \not\vdash_\KB \phi$. In this case, the set $\Phi\cup\{\neg\phi\}$ can be extended to a maximal \KB-consistent set $\Delta^+$. In the canonical model for \ls, denoted \CM, we have $\CM,\Delta^+\models \chi$ for any formula $\chi\in \Phi\cup\{\neg\phi\}$. This conclusion leads to $\Phi \not\models_\ls \phi$.
\end{proof}

\subsubsection{Proof by constructing a standard model (a path-based canonical model)}
\label{sec:completeness3}

When dealing with logics that involve distributed and/or mutual knowledge, the traditional canonical model method proves to be ineffective. To address this, a classical method has been proposed for logics with distributed knowledge (cf. \cite{FHV1992}, whose approach is based on unraveling techniques dating back to \cite{Sahlqvist1975}). The method starts by treating distributed and mutual knowledge as individual knowledge, then constructs a pseudo model incorporating these elements. This pseudo model is subsequently unraveled into a path-based tree-like model and then identified/folded into the required model. A simplified approach suggested by \cite{WA2020} advocates for directly building a path-based tree-like model, termed a standard model, bypassing the actual process of unraveling and identification/folding. Our logics can also adopt this method, and we aim to construct a standard model to achieve the completeness result.

\paragraph{\underline{Completeness of \KD}}

\begin{definition}[canonical path for \ld]
A \emph{canonical path} for \ld is defined as a sequence $\langle \Phi_0, G_1, \Phi_1, \dots, G_n, \Phi_n \rangle$, where:
\begin{itemize}
\item $\Phi_0, \Phi_1, \dots, \Phi_n$ represent maximal \KD-consistent sets of \langd-formulas,
\item $G_1, \dots, G_n$ denote groups of agents, i.e., nonempty subsets of \ag.
\end{itemize}
In the context of a canonical path (this also applies to canonical paths defined later) $s = \langle \Phi_0,G_1,\Phi_1,\dots,G_n,\Phi_n \rangle$, we denote $\Phi_n$ as $tail(s)$.
\qed
\end{definition}

\begin{definition}[standard model for \ld]
The standard model for \ld is represented as the tuple $\CM = (\CW, \CA, \CE, \CC, \CV)$, where:
\begin{itemize}
\item \CW is the set of all canonical paths for \ld;
\item $\CA = \wp(\ag)$;
\item $\CE : \CW \times \CW \to \wp(\CA)$ is defined such that for any $s, t \in \CW$, 
$$\CE(s,t)  = \left\{\begin{array}{ll}
	\bigcup_{a\in G} \CC(a), & \text{if $t$ is $s$ extended with $\langle G,\Psi \rangle$ and $\{\chi\mid D_G\chi \in tail(s)\} \subseteq \Psi$},\\
	\emptyset, & \text{otherwise;}
\end{array}\right.$$
\item $\CC : \ag \to \wp(\CA)$ is defined such that for any agent $a$, $\CC(a)=\{G \in \ag \mid a \in G \}$;
\item $\CV: \CW \to \wp(\prop)$ is defined such that for any $s \in \CW$, $\CV (s) = \{p \in \prop \mid p \in tail(s)\}$.
\qed
\end{itemize}
\end{definition}

\begin{lemma}[standardness]
The standard model for \ld, as defined above, indeed qualifies as a model.
\qed
\end{lemma}

The above definitions and lemma form a groundwork to build upon for subsequent proofs and theorems, providing a robust framework that can be applied to complex logics involving distributed and mutual knowledge.

\begin{lemma}[Truth Lemma]
\label{lem:truth-eld}
In the canonical model for \ld, represented as $\CM = (\CW,\CA,\CE,\CC,\CV)$, the following correspondence holds: for any $s \in \CW$ and any \langd-formula $\phi$, $\phi \in tail(s)$ if and only if $\CM,s \models_{\ld} \phi$.
\end{lemma}
\begin{proof}
We only show two cases here, namely, when $\phi = K_a\psi$ and when $\phi = D_G\psi$.

Case $\phi = K_a \psi$:
Suppose $K_a\psi\in tail(s)$, but $\CM,s\not\models_\ld K_a\psi$. Then, there exists some $t\in \CW$ such that $\CC(a)\subseteq \CE (s,t)$ and $\CM,t\not\models_\ld \psi$. This implies that for some group $G$ that includes $a$, the set $\{\chi\mid D_G\chi\in tail(s)\}$ is a subset of $tail(t)$. 
Hence, we have $\psi\in tail(t)$ since $K_a\psi \in tail(s)$ implies $D_G\psi \in tail(s)$. However, the induction hypothesis suggests that $\CM,t\models_\ld \psi$, leading to a contradiction.
In the other direction, suppose $K_a\psi \notin tail(s)$, but $\CM,s\models_\ld K_a\psi$. Then, for any $t\in \CW$, $\CC(a)\subseteq \CE (s,t)$ implies $\CM,t\models_\ld \psi$. If we extend the set $\{\neg\psi\}\cup\{\chi\mid K_a\chi\in\Gamma\}$ to a maximal \KD-consistent set $\Delta^+$, we find that $\CC(a)\subseteq \CE(s,t)$, where $t$ extends $s$ with $\langle\{a\},\Delta^+\rangle$. However, the induction hypothesis suggests that $\CM,t\models_\ld \neg\psi$, again leading to a contradiction.

Case $\phi = D_G\psi$:
Suppose $D_G\psi\in tail(s)$, but $\CM,s\not\models_\ld D_G\psi$. Then there exists some $t\in \CW$ such that $\bigcup_{a\in G}\CC(a)\subseteq \CE (s,t)$ and $\CM,t\not\models_\ld \psi$. This implies that for some group $H$ that is a superset of $G$, the set $\{\chi\mid D_H\chi\in tail(s)\}$ is a subset of $tail(t)$. Hence, we have $\psi \in tail(t)$ since $D_G\psi\in tail(s)$ implies $D_H\psi\in tail(s)$. However, the induction hypothesis suggests that $\CM,t\models_\ld \psi$, leading to a contradiction.
In the other direction, suppose $D_G\psi \notin tail(s)$, but $\CM,s\models_\ld D_G\psi$. Then, for any $t\in \CW$ $\bigcup_{a\in G}\CC(a)\subseteq \CE (s,t)$ implies $\CM,t\models_\ld \psi$. If we extend the set $\{\neg\psi\}\cup\{\chi\mid D_G\chi\in\Gamma\}$ (whose consistency can be proven similarly) to a maximal \KD-consistent set $\Delta^+$, we find that $\bigcup_{a\in G}\CC(a)\subseteq \CE(s,t)$, where $t$ extends $s$ with $\langle G,\Delta^+\rangle$. However, the induction hypothesis suggests that $\CM,t\models_\ld \neg\psi$, again leading to a contradiction.
\end{proof}

We can now leverage the above proof to establish the completeness of \KD in a similar manner.

\begin{theorem}[completeness of \KD, with a direct proof]
For any \langd-formula $\phi$ and any set $\Phi$ of \langd-formulas, if $\Phi \models_\ld \phi$, then $\Phi \vdash_\KD \phi$.
\qed
\end{theorem}
\begin{proof}
Assuming $\Phi \nvdash_\KD \phi$, then $\Phi\cup\{\neg\phi\}$ is \KD-consistent. As a result, it can be extended to a maximal \KD-consistent set $\Delta^+$. Given the standard model $\CM$ for \ld, for any formula $\phi \in \Delta$, we can conclude that $\CM, \langle \Delta^+ \rangle \models \psi$. This leads us to infer that $\Phi \not\models_\ld \phi$.
\end{proof}

\paragraph{\underline{Completeness of \KBD}}

By revising the standard model, we can obtain a completeness proof for \KBD. The process is similar to the proof for \KD. The main difference is that we need to adjust the standard model for it to be symmetric. Canonical paths for \lsd are defined in the same way as those for \ld, with the only difference being the need to replace ``\KD-consistent sets'' to ``\KBD-consistent sets''.

\begin{definition}[standard model for \lsd]\label{def:cm-elds}
The \emph{standard model} for \lsd is obtained by adapting the standard model for \ld (Definition~\ref{def:cm-el}) in two places:
(1) \CW is changed to be the set of all canonical paths for \lsd, and (2)
$$\CE(s,t)  = \left\{\begin{array}{ll}
	\bigcup_{a\in G} \CC(a), & \text{if $t$ is $s$ extended with $\langle G,\Psi \rangle$, $\{\chi\mid D_G\chi \in tail(s)\} \subseteq \Psi$ and $\{\chi\mid D_G\chi \in \Psi\} \subseteq tail(s)$},\\
	\bigcup_{a\in G} \CC(a), & \text{if $s$ is $t$ extended with $\langle G,\Psi \rangle$, $\{\chi\mid D_G\chi \in tail(t)\} \subseteq \Psi$ and $\{\chi\mid D_G\chi \in \Psi\} \subseteq tail(t)$},\\
	\emptyset, & \text{otherwise.}
\end{array}\right.$$
\end{definition}

\begin{lemma}[standardness]
The standard model for \lsd is a similarity model.
\end{lemma}
\begin{proof}
Note that $\emptyset \notin \CC(a)$ for any agent $a$. This implies that for any $s,t \in \CW$, $\CE(s,t) \neq \CA$, thereby meeting the criterion of positivity. Additionally, the condition of symmetry is fulfilled as $\CE$ is a commutative function.
\end{proof}

\begin{lemma}[Truth Lemma]
Let $\CM = (\CW,\CA,\CE,\CC,\CV)$ be the canonical model for \lsd. For any $s \in \CW$ and \langd-formula $\phi$, we have $\phi \in tail(s)$ iff $\CM,s \models_{\lsd} \phi$.
\end{lemma}
\begin{proof}
The proof proceeds by considering the two directions of the equivalence separately. The direction from $\phi \in tail(s)$ to $\CM,s\models_\lsd \phi$ can be shown similarly to the proof of Lemma~\ref{lem:truth-eld}. For the reverse direction, we consider the cases $\phi = K_a\psi$ and $\phi=D_G\psi$ here.

Case $\phi=K_a\psi$: This case can be shown similarly to Lemma~\ref{lem:truth-KB}.

Case $\phi = D_G\psi$: This case is handled by contradiction.
Suppose $D_G\psi\notin tail(s)$, but $\CM,s\models_\lsd D_G\psi$. Then, by the definition of \CE, $\bigcup_{a\in G}\CC(a)\subseteq \CE (s,t)$ implies $\CM,t\models_\lsd \psi$ for any $t\in \CW$. We can extend $\{\neg\psi\}\cup\{\chi\mid D_G\chi\in\Gamma\}\cup\{\neg D_G\neg\chi\mid \chi\in\Gamma\}$ to some maximal \KBD-consistent set $\Delta^+$. Then, we have $\bigcup_{a\in G}\CC(a)\subseteq \CE(s,t)$ where $t$ extends $s$ with $\langle G,\Delta^+\rangle$. By the induction hypothesis we have $\CM,t\models_\ld \neg\psi$, leading to a contradiction.
\end{proof}

By applying the Truth Lemma, we can now prove the completeness theorem for \lsd.

\begin{theorem}[completeness of \KBD, part 3]
\label{thm:completeness3}
For any \langd-formula $\phi$ and any set $\Phi$ of \langd-formulas, if $\Phi \models_\lsd \phi$, then $\Phi \vdash_\KBD \phi$.
\qed
\end{theorem}

\paragraph{\underline{Completeness of \KM and \KBM}}

The completeness of \KM and \KBM can be demonstrated in a manner that parallels the completeness of \KD and \KBD. While we will not delve into the intricate details of the proofs, we will outline the necessary adaptations to the definitions of the standard model for each of the logics.

A canonical path for \lm or \lsm mirrors that for \ld. The only modification required is the adjustment of the maximal consistent sets to align with the axiomatic system being considered.

When defining the standard model for \lm, we substitute \CW with the set of all canonical paths for \lm, and let
$$\CE(s,t)  = \left\{\begin{array}{ll}
	\bigcap_{a\in G} \CC(a), & \text{if $t$ is an extension of $s$ with $\langle G,\Psi \rangle$ and $\{\chi\mid M_G\chi \in tail(s)\} \subseteq \Psi$},\\
	\emptyset, & \text{otherwise.}
\end{array}\right.$$
Similarly, while forming the standard model for \lsm, we replace \CW with the set of all canonical paths for \lsm, and let
$$\CE(s,t)  = \left\{\begin{array}{ll}
	\bigcap_{a\in G} \CC(a), & \text{if $t$ is $s$ extended with $\langle G,\Psi \rangle$, $\{\chi\mid M_G\chi \in tail(s)\} \subseteq \Psi$ and $\{\chi\mid M_G\chi \in \Psi\} \subseteq tail(s)$},\\
	\bigcap_{a\in G} \CC(a), & \text{if $s$ is $t$ extended with $\langle G,\Psi \rangle$, $\{\chi\mid M_G\chi \in tail(t)\} \subseteq \Psi$ and $\{\chi\mid M_G\chi \in \Psi\} \subseteq tail(t)$},\\
	\emptyset, & \text{otherwise.}
\end{array}\right.$$
Please note that $\bigcup_{a\in G} \CC(a) = \{ H \mid H \cap G \neq \emptyset \}$, which includes all the common epistemic abilities of those groups $H$ that intersects with $G$. Additionally, $\bigcap_{a\in G} \CC(a) = \{ H \mid G \subseteq H \}$, which represents all the supersets of $G$.

By using analogous proof structures, we can demonstrate the standardness of these models, derive a Truth Lemma, and subsequently establish the completeness of the logics.

\begin{theorem}[completeness, part 4]\label{thm:completeness4}
\KM and \KBM are strongly complete for \lm and \lsm, respectively.
\qed
\end{theorem}

\subsubsection{Incorporation of both distributed and mutual knowledge}
\label{sec:completeness4}

We now discuss logics and their axiomatic systems that incorporate both distributed and mutual knowledge but exclude common knowledge. Specifically, we focus on \KDM and \KBDM. The construction process requires careful consideration of the intricate interaction between the two types of knowledge modalities.

\paragraph{\underline{Completeness of \KDM and \KBDM}}

\begin{definition}[canonical path for \ldm/\lsdm]
A \emph{canonical path for \ldm} is a sequence $\langle \Phi_0, I_1, \Phi_1, \dots, I_n, \Phi_n \rangle$ where:
\begin{itemize}
\item $\Phi_0,\Phi_1,\dots,\Phi_n$ are maximal \KDM-consistent sets of \langdm-formulas;
\item $I_1, \dots, I_n $ are of the form $(G, d)$ or $(G, m)$, with $G$ denoting a group, and ``$d$'' and ``$m$'' being just two distinct characters.
\end{itemize}
The \emph{canonical path for \lsdm} is similarly defined , with the only alteration being the replacement of \KDM with \KBDM.
\qed
\end{definition}

\begin{definition}[standard model for \ldm/\lsdm]
\label{def:sm-ldm}
The standard model for \ldm is a tuple $\CM = (\CW, \CA, \CE, \CC, \CV)$ where:
\begin{itemize}
\item \CW is the set of all canonical paths for \ldm;
\item $\CA = \wp(\ag)$;
\item $\CE : \CW \times \CW \to \wp(\CA)$ is such that for any $s, t \in \CW$, 
$$\CE(s,t)  = \left\{\begin{array}{ll}
	\bigcup_{a\in G} \CC(a), & \text{if $t$ is $s$ extended with $\langle (G,d), \Psi \rangle$ and $\{\chi\mid D_G\chi \in tail(s)\} \subseteq \Psi$},\\
	\bigcap_{a\in G} \CC(a), & \text{if $t$ is $s$ extended with $\langle (G,m), \Psi \rangle$ and $\{\chi\mid M_G\chi \in tail(s)\} \subseteq \Psi$},\\
	\emptyset, & \text{otherwise;}
\end{array}\right.$$
\item $\CC : \ag \to \wp(\CA)$ is such that for any agent $a$, $\CC(a)=\{G \in \ag \mid a \in G \}$;
\item $\CV: \CW \to \wp(\prop)$ is such that for any $s \in \CW$, $\CV (s) = \{p \in \prop \mid p \in tail(s)\}$.
\end{itemize}
The standard model for \lsdm is largely defined in the same way, with some changes:
\begin{itemize}
\item \CW is the set of all canonical paths for \lsdm;
\item $\CE : \CW \times \CW \to \wp(\CA)$ is defined with the necessary adjustments to account for symmetry, namely, for any $s, t \in \CW$, 
\end{itemize}
$$\CE(s,t)  = \left\{\begin{array}{ll}
	\bigcup_{a\in G} \CC(a), & \text{if $t$ extends $s$ with $\langle (G,d), \Psi \rangle$,  $\{\chi\mid D_G\chi \in tail(s)\} \subseteq \Psi$ and $\{\chi\mid D_G\chi \in \Psi \} \subseteq tail(s)$},\\
	\bigcup_{a\in G} \CC(a), & \text{if $s$ extends $t$ with $\langle (G,d), \Psi \rangle$,  $\{\chi\mid D_G\chi \in tail(t)\} \subseteq \Psi$ and $\{\chi\mid D_G\chi \in \Psi \} \subseteq tail(t)$},\\
	\bigcap_{a\in G} \CC(a), & \text{if $t$ extends $s$ with $\langle (G,m), \Psi \rangle$, $\{\chi\mid M_G\chi \in tail(s)\} \subseteq \Psi$ and $\{\chi\mid M_G\chi \in \Psi \} \subseteq tail(s)$},\\
	\bigcap_{a\in G} \CC(a), & \text{if $s$ extends $t$ with $\langle (G,m), \Psi \rangle$, $\{\chi\mid M_G\chi \in tail(t)\} \subseteq \Psi$ and $\{\chi\mid M_G\chi \in \Psi \} \subseteq tail(t)$},\\
	\emptyset, & \text{otherwise.}
\end{array}\right.$$
\end{definition}

It is straightforward to verify that the standard model for \ldm is indeed a model, and the standard model for \lsdm is a similarity model.

\begin{lemma}[Truth Lemma]\label{lem:truthdm}
The following statements hold:
\begin{enumerate}
\item\label{it:truthdm1} Let $\CM = (\CW, \CA, \CE, \CC, \CV)$ be the standard model for \ldm. For any $s \in \CW$ and any \langdm-formula $\phi$, $\phi \in tail(s)$ if and only if $\CM,s \models_{\ldm} \phi$;
\item\label{it:truthdm2} Let $\CM = (\CW, \CA, \CE, \CC, \CV)$ be the standard model for \lsdm. For any $s \in \CW$ and any \langdm-formula $\phi$, $\phi\in tail(s)$ if and only if $\CM,s\models_{\lsdm}\phi$.
\end{enumerate}
\end{lemma}
\begin{proof}
We demonstrate only the first clause here. The second clause can be proven in a similar manner. The proof is once more by induction on $\phi$, and here we only display the case when $\phi=M_G\psi$.

Suppose $M_G\psi\in tail(s)$, but $\CM,s\not\models_\ldm M_G\psi$, then there exists some $t\in \CW$ such that $\bigcap_{a\in G}\CC(a)\subseteq \CE (s,t)$ and $\CM,t\not\models_\ldm \psi$. Therefore, $\{\chi\mid M_H\chi\in tail(s)\}\subseteq tail(t)$ for some group $H$ such that $H\subseteq G$ or $\{\chi\mid D_J\chi\in tail(s)\}\subseteq tail(t)$ for some group $J$ such that $G\cap J\neq\emptyset$. In both scenarios, we have $\psi\in tail(t)$ since $M_G\psi\in tail(s)$ implies $M_H\psi, D_J\psi\in tail(s)$. By the induction hypothesis, we have $\CM,t\models_\ldm \psi$, which leads to a contradiction.
Suppose $M_G\psi\notin tail(s)$, but $\CM,s\models_\ldm M_G\psi$, then $\bigcap_{a\in G}\CC(a)\subseteq \CE (s,t)$ implies $\CM,t\models_\ldm \psi$ for any $t\in \CW$. Extend $\{\neg\psi\}\cup\{\chi\mid M_G\chi\in\Gamma\}$ to some maximal \KDM-consistent set $\Delta^+$, thus $\bigcap_{a\in G}\CC(a)\subseteq \CE(s,t)$ where $t$ extends $s$ with $\langle (G,m),\Delta^+\rangle$. By the induction hypothesis, we have $\CM,t\models_\ldm \neg\psi$, leading to a contradiction.
\end{proof}

\begin{theorem}[completeness, part 5]\label{thm:completeness5}
The following hold:
\begin{enumerate}
\item \KDM is strongly complete for \ldm;
\item \KBDM is strongly complete for \lsdm.
\qed
\end{enumerate}
\end{theorem}

\subsubsection{Proof by a finitary standard model}
\label{sec:completeness5}

We will now delineate the extension of the completeness results to the rest of the logics with common knowledge, deploying a finitary method for this purpose. We can only achieve weak completeness due to the non-compact nature of the common knowledge modality. To prove the completeness of logics with common knowledge, we often also need to address the modality for distributed or mutual knowledge.

In this section, we focus on providing the completeness proofs for \KCDM and \KBCDM. By making simple adaptations, we can obtain the completeness of the axiomatic systems for their sublogics with common knowledge. We adapt the definition of the \emph{closure} of a formula presented in \cite{WA2020} , to cater to formulas with modalities $D_G$ and/or $M_G$.

\begin{definition}\label{def:cl}
For an \langcdm-formula $\phi$, we define $cl(\phi)$ as the minimal set satisfying the subsequent conditions:
\begin{enumerate}
\item\label{it:cl-id} $\phi\in cl(\phi)$;
\item\label{it:cl-sub} if $\psi$ is in $cl(\phi)$, so are all subformulas of $\psi$;
\item\label{it:cl-neg} $\psi\in cl(\phi)$ implies ${\sim}\psi\in cl(\phi)$, where $\NEG\psi=\neg\psi$ if $\psi$ is not a negation and $\NEG\psi=\chi$ if $\psi=\neg\chi$;
\item\label{it:cl-1} $K_a\psi\in cl(\phi)$ implies $D_{\{a\}}\psi, M_{\{a\}}\psi\in cl(\phi)$;
\item\label{it:cl-d1} $D_{\{a\}}\psi\in cl(\phi)$ implies $K_a\psi\in cl(\phi)$;
\item\label{it:cl-d2} For groups $G$ and $H$, if $G \subseteq H$ and $H$ appears in $\phi$, then $D_G\psi\in cl(\phi)$ implies $D_H\psi\in cl(\phi)$;
\item\label{it:cl-c1} $C_G\psi\in cl(\phi)$ implies $\{ K_a\psi, K_a C_G\psi \mid a \in G \} \subseteq cl(\phi)$;
\item\label{it:cl-m1} $M_G\psi\in cl(\phi)$ implies $\{ K_a\psi \mid a \in G\} \subseteq cl(\phi)$;
\item\label{it:cl-m2} For groups $G$ and $H$, if $H \subseteq G$ and $H$ appears in $\phi$, then $M_G\psi\in cl(\phi)$ implies $M_H\psi\in cl(\phi)$.
\qed
\end{enumerate}
\end{definition}
Given that there are finitely many groups appearing in $\phi$, and every group comprises only a finite number of agents, we can readily confirm that $cl(\phi)$ is finite for any given formula $\phi$.

Subsequently, we introduce the concept of a \emph{maximal consistent set of formulas within a closure}. For a comprehensive definition, which is naturally contingent on the specific axiomatic system under consideration, we refer to established literature, for example, \cite{vDvdHK2008}.

\begin{definition}[canonical path for \lcdm/\lscdm in a closure]
Given an \langcdm-formula $\phi$, we define a \emph{canonical path for \lcdm in $cl(\phi)$} as a sequence $\langle \Phi_0, I_1, \Phi_1, \dots, I_n, \Phi_n \rangle$ that satisfies the following conditions:
\begin{itemize}
\item $\Phi_0,\Phi_1,\dots,\Phi_n$ are maximal \KCDM-consistent sets of \langcdm-formulas in $cl(\phi)$;
\item $I_1, \dots, I_n $ take the form $(G, d)$ or $(G, m)$, where $G$ is a group, and ``$d$'' and ``$m$'' are simply two distinct letters.
\end{itemize}
We can define the \emph{canonical path for \lscdm} similarly by only substituting \KCDM with \KBCDM in the above.
\end{definition}

Given an \langcdm-formula $\phi$, we can construct the \emph{standard model for \lcdm with respect to $cl(\phi)$} in a manner that closely mirrors the construction of the standard model for \ldm (as per Definition~\ref{def:sm-ldm}). The primary differences lie in bounding the canonical paths by the closure and adjusting the logics accordingly. More specifically, we need to (1) replace all occurrences of ``\langdm'' with ``\langcdm'', and ``\ldm'' with ``\lcdm''; (2) within the definition of \CW, replace ``canonical paths for \ldm'' with ``canonical paths for \lcdm in $cl(\phi)$''. In a similar vein, we can modify the standard model for \lscdm with respect to $cl(\phi)$ from that of \lsdm (Definition~\ref{def:sm-ldm}). Furthermore, it is straightforward to confirm that the standard model for \lcdm (in any closure of a given formula) is indeed a model, and that for \lscdm constitutes a similarity model.

\begin{lemma}[Truth Lemma]\label{lem:truthcdm}
Given an \langcdm-formula $\theta$,
\begin{enumerate}
\item\label{it:truthcdm1}Let $\CM = (\CW, \CA, \CE, \CC, \CV)$ be the standard model for \lcdm with respect to $cl(\theta)$, for any $s \in \CW$ and \langcdm-formula $\phi \in cl(\theta)$, we have $\phi \in tail(s)$ iff $\CM,s \models_{\lcdm} \phi$;
\item\label{it:truthcdm2}Let $\CM = (\CW, \CA, \CE, \CC, \CV)$ be the standard model for \lscdm with respect to $cl(\theta)$, for any $s \in \CW$ and \langcdm-formula $\phi \in cl(\theta)$, we have $\phi \in tail(s)$ iff $\CM,s \models_{\lscdm} \phi$.
\end{enumerate}
\end{lemma}
\begin{proof}
\ref{it:truthcdm1}. We show the lemma by induction on $\phi$. We omit the straightforward cases here. The cases involving modalities do not significantly differ from those in previous proofs of the truth lemmas. However, attention must be given to handle the closure appropriately. Our focus here will be primarily on the cases concerning the common knowledge operators.

Suppose $C_G\psi\in tail(s)$, but $\CM,s\not\models_\lcdm C_G\psi$, then there are $s_i\in \CW$, $a_i\in G$, $0\leq i\leq n$ for some $n\in\mbN$ such that: $s_0=s$, $\CM,s_n\not\models\psi$ and $\CC(a_i)\subseteq \CE (s_{i-1},s_i)$ for $1\leq i\leq n$. Since $\CC(a_i)\subseteq \CE (s_{i-1},s_i)$, we have either $\{\chi\mid D_H\chi\in tail(s_{i-1})\}\subseteq tail(s_i)$ for some $H$ containing $a_i$ or $\{\chi\mid M_{\{a_i\}}\chi\in tail(s_{i-1})\}\subseteq tail(s_i)$. In both cases, $\{\chi\mid K_{a_i}\chi\in tail(s_{i-1})\}\subseteq tail(s_i)$. Since $C_G\psi\in tail(s_i)$ implies $K_{a_i}C_G\psi,K_{a_i}\psi\in tail(s_i)$, we can infer that $C_G\psi,\psi\in tail(s_n)$. Then by the induction hypothesis, $\CM,s_n\models_\lcdm \psi$, leading to a contradiction.

Suppose $C_G\psi\not\in tail(s)$, but $\CM,s\models_\lcdm C_G\psi$. Thus for any $s_i\in \CW$ and $a_i\in G$ where $0\leq i\leq n$, such that: $s_0=s$ and $\CC(a_i)\subseteq\CE(s_{i-1},s_i)$, we have $M,s_n\models_\lcdm \psi$ and $M,s_n\models_\lcdm C_G\psi$. Collect all such possible $s_n$ above and $s$ into the set $\mcS$; similarly collect all the $tail(s_n)$ and $tail(s)$ into the set $\Theta$. We define $\delta=\bigvee_{t\in\mcS}\widehat{tail(t)}$, where for any $t \in \CW$, $\widehat{tail(t)}$ stands for $\bigwedge tail(t)$. (In general, for any finite set $\Psi$ of formulas, write $\widehat{\Psi}$ for $\bigwedge\Psi$.)
We claim that $\vdash_{\lcdm}\delta\ra K_a\delta$ and $\vdash_{\lcdm}\delta\ra K_a\psi$ for any $a\in G$. By this claim and (C2) we have $\vdash_{\lcdm}\delta\ra C_G\psi$, and then by $\widehat{\Gamma}\ra\delta$ we have $\widehat{\Gamma} \ra C_G\psi$. In this way we obtain $C_G\psi\in\Gamma$, which leads to a contradiction. As for the proof of the claim:

(1) Suppose $\nvdash_{\lcdm}\delta\ra K_a\delta$, then $\delta\wedge \neg K_a\delta$ is consistent. Then there exists $t_0\in\mcS$ such that $\widehat{tail(t_0)}\wedge\neg K_a\delta$ is consistent. Notice that $\vdash_{\lcdm}\bigvee_{t\in \CW}\widehat{tail(t)}$, hence we have a consistent set $\widehat{tail(t_0)}\wedge\neg K_a\neg \widehat{tail(t_1)}$ for some $t_1\in \CW\setminus\mcS$ such that $tail(t_1)\notin\Theta$. Thus we have $\{\chi\mid K_a\chi\in tail(t_0)\}\subseteq tail(t_1)$, which implies $\{\chi\mid D_{\{a\}}\chi\in tail(t_0)\}\subseteq tail(t_1)$. Now we let $t_2$ be $t_0$ extended with $\langle(\{a\},d),tail(t_1)\rangle$, we have $\CC(a)\subseteq\CE(t_0,t_2)$. Hence $t_2\in\mcS$ but $tail(t_2)=tail(t_1)\notin\Theta$, a contradiction!

(2) Suppose $\nvdash_{\lcdm}\delta\ra K_a\psi$, then $\delta\wedge \neg K_a\psi$ is consistent. So there exists $t_0\in\mcS$ such that $\widehat{tail(t_0)}\wedge\neg K_a\psi$ is consistent. Thus $\{\NEG\psi\}\cup\{\chi\mid K_a\chi\in tail(t_0)\}$ is consistent as before. Hence it can be extended to some some max consistent subset $\Delta^+$ on $cl(\theta)$. Let $t_1$ be $t_0$ extended with $\langle(\{a\},d),\Delta^+\rangle$, we have $\CC(a)\subseteq\CE(t_0,t_1)$. Hence $t_1\in\mcS$ and then $\CM,t_1\models_\lcdm\psi$, which contradicts with $\NEG\psi\in tail(t_1)$ by the induction hypothesis.

The second clause can be shown in a quite similar way to the above. Details are omitted here.
\end{proof}

\begin{theorem}[completeness, part 6]\label{thm:completeness6}
The following hold:
\begin{enumerate}
\item \KCDM is weakly complete for \lcdm, or equivalently, every \KCDM-consistent formula is satisfiable;
\item \KBCDM is weakly complete for \lscdm, or equivalently, every \KBCDM-consistent formula is s-satisfiable.
\end{enumerate}
\end{theorem}
\begin{proof}
1. Consider an \langcdm-formula $\phi$ that is \KCDM consistent. This formula can be augmented to a maximal \KCDM-consistent subset $\Delta^+$ of $cl(\phi)$. By applying the Truth Lemma, we find that for the standard model \CM for \lcdm with respect to $cl(\phi)$, we have $\CM, \langle \Delta^+ \rangle \models \phi$.

2. Similarly, if we take any \langcdm-formula $\phi$ that is \KBCDM consistent, it is satisfied in the standard model for \lscdm with respect to $cl(\phi)$. This standard model is a similarity model. 
\end{proof}

As stated earlier in this section, the completeness proofs for \KC, \KBC, \KCD, \KBCD, \KCM and \KBCM, can be adapted from the proofs for \KCDM or \KBCDM as provided above. This can be achieved by employing the finitary method, in conjunction with the approach we adopted for the axiomatic systems for logics lacking common knowledge.

Within this set, the completeness proofs for \KC and \KCD have already been demonstrated in Section~\ref{sec:completeness1} through the method of translation. Hence, we list the completeness results for the remaining four systems below.

\begin{theorem}[completeness, part 7]
The following statements hold:
\begin{enumerate}
\item \KBC is weakly complete for \lsc;
\item \KCM is weakly complete for \lcm;
\item \KBCM is weakly complete for \lscm;
\item \KBCD is weakly complete for \lscd.
\qed
\end{enumerate}
\end{theorem}

\section{Computational Complexity}
\label{sec:complexity}

In this section we study the computational complexity of the model checking problem and the satisfiability/validity problem, for all the logics that have been proposed in Section~\ref{sec:logics}. The \emph{model checking problem} for a logic is to determine, for a given formula $\phi$, a finite (similarity) model $M$ for the logic and a state $s$ of the (similarity) model, whether $M, s \models \phi$. The \emph{satisfiability problem} for a logic is to determine whether a given formula $\phi$ is (s-)satisfiable or not in the logic. Its dual problem -- the \emph{valid problem} -- is to decide whether a given $\phi$ is (s-)satisfiable; that is, whether $\neg \phi$ is unsatisfiable. We shall focus mainly on the model checking problem and satisfiability problem, but since our results only concern the classes PSPACE and EXPTIME, the complexity of validity problem is either co-PSPACE (the complement of PSPACE) or co-EXPTIME (the complement of EXPTIME), which are known to be equivalent to PSPACE and EXPTIME, respectively. As a result, it does not make any difference here between the complexity of satisfiability and validity problems. So we speak of the validity problem as well from time to time.

\subsection{The input}

We define the measure of the input. The \emph{length} of a formula $\phi$, denoted $|\phi|$, is defined to be the number of symbols that occur in $\phi$ (including the symbols for brackets), just as in \cite[Section 3.1]{FHMV1995}; or more precisely defined inductively by the structure of $\phi$, i.e., when $\phi$ is:
\begin{itemize}
\item Propositional variable $p$: $|p| = 1$
\item Negation $\neg\psi$: $| \neg \psi | = |\psi| +1$
\item Implication $(\psi \ra \chi)$: $| (\psi \ra \chi) | = |\psi| + |\chi| + 3$
\item Individual knowledge $K_a \psi$: $|K_a\psi| = |\psi| + 2$
\item Group knowledge: $|C_G\psi| = |\psi| + 2|G| + 2$,  and similarly for $D_G \psi$ and $M_G\psi$. E.g., $| (p \ra C_{\{a,b,c\}} q) | = 13$.
\end{itemize}

The \emph{size} of a model $M = (W,\ab,E,C,\nu)$, denoted $|M|$, with respect to a given formula $\phi$, is the sum of the following:
\begin{itemize}
\item $|W|$, i.e., the cardinality of the domain,
\item $|\ab|$, i.e., the cardinality of the set of epistemic abilities,
\item $|E|$, defined to be $|W|^2 \cdot |\ab|$,
\item $|C|$ w.r.t. $\phi$, which is defined as $|\phi| \cdot |A|$, and
\item $|\nu|$ w.r.t. $\phi$, defined to be $|W| \cdot |\phi|$.
\end{itemize}

Finally, given a formula $\phi$ and a model $M$ (with a designated state $s$ of it), the \emph{size of the input} is $|\phi| + |M| + 1$.

\subsection{The model checking problem}
\label{sec:mc}

We first show that the model checking problem for \l is in P, and then extend the result to all other logics.

\begin{lemma}\label{lem:mc-el}
The model checking problem for \l and \ls are both in P.
\end{lemma}
\begin{proof}
Given a model $M=(W,\ab,E,C,\nu)$, a state $s \in W$ and a formula $\phi$, we need to decide whether $M,s \models \phi$. In order to do so, we present an algorithm (Algorithm~\ref{alg:val}) for calculating $Val(M,\phi)$, the truth set of $\phi$ in $M$, i.e., $\{ s \in W \mid M,s \models \phi\}$. The question about whether $M,s \models \phi$ holds is thus reduced to the membership testing in $Val(M,\phi)$, which takes at most $|W|$ steps in addition to the time costs on computing $Val(M,\phi)$.%

\begin{algorithm}
\caption{Function $Val(M,\phi)$: computing the truthset in \l and \ls}\label{alg:val}
\small
\begin{algorithmic}[1]
\Require a weighted model $M = (W, \ab, E, C, \nu)$ and a formula $\phi$
\Ensure $\{ s \mid M,s\models\phi \}$
\Initialize{$tmpVal \gets \emptyset$}

\If{$\phi=p$} \Return $\{ s \in W \mid p \in \nu(s)\}$

\ElsIf{$\phi=\neg\psi$} \Return $W \setminus Val(M,\psi)$

\ElsIf{$\phi=\psi\to\chi$} \Return $(W \setminus Val(M,\psi)) \cup Val(M,\chi)$

\ElsIf{$\phi=K_a\psi$}
 \ForAll{$t\in W$}
  \Initialize{$n \gets \True$}
  \ForAll{$u\in W$}
    \If{$C(a)\subseteq E(t,u)$ \algand $u\not\in Val(M,\psi)$} {$n \gets \False$}
    \EndIf
  \EndFor
  \If{$n=\True$} $tmpVal \gets tmpVal \cup \{t\}$
  \EndIf
  \EndFor
 \State \Return $tmpVal$ \Comment{This returns $\{t\in W\mid \forall u\in W: C(a)\subseteq E(t,u)\Rightarrow u\in Val(M,\psi)\}$}
\EndIf
\end{algorithmic}
\end{algorithm}

It is not hard to verify that $Val(M,\phi)$ is indeed the set of states of $M$ at which $\phi$ is true. In particular, in the case for the $K_a$ operator,
\[
\begin{array}{lll}
M,s\models K_a\psi & \iff & \forall u \in W: C(a)\subseteq E(s,u) \Rightarrow M,u\models\psi \\
 & \iff & \forall u \in W: C(a)\subseteq E(s,u) \Rightarrow u \in Val(M,\psi) \hfill \text{(IH)}\\
 & \iff & s \in \{t \in W \mid \forall u\in W: C(a)\subseteq E(t,u)\Rightarrow u\in Val(M,\psi)\} 
\end{array}
\]

The cost for computing $Val(M,\phi)$ is in polynomial time. In the case for $K_a \psi$ -- the most time-consuming case here -- there are two while-loops over $W$, and checking $C(a) \subseteq E(t,u)$ costs at most $|A|$ steps, and the membership checking $u \notin Val(M,\psi)$ takes at most $|W|$ steps; this costs $|W|^2 \cdot (|A| + |W|)$. The algorithm for computing $Val(M,\phi)$ calls itself recursively, but only for a subformula of $\phi$, and the maximum number of recursion is bounded by $|\phi|$, i.e., the length of $\phi$. So the total time cost for computing $Val(M,\phi)$ is $|W|^2 \cdot (|A| + |W|) \cdot |\phi|$. Considering the input size, we find that the total time cost is within $O(n^2)$. So the lemma holds.

The model checking problem for \ls is a subproblem of that for \l, and hence in P as well.
\end{proof}

Now we extend the above result to the other logics (Theorem~\ref{thm:complexity}). Before that we come up with a definition and lemmas that support it.

\begin{definition}\label{def:trans-e}
Given a model $M = (W,\ab,E,C,\nu)$ and a formula $\phi$, let
\begin{itemize}
\item $\ab_\phi=\ab\cup\{G \mid\text{``$E_G$'' or ``$C_G$'' appear in $\phi$}\}$,
\item for all states $s,t\in W$, $E_\phi(s,t)=E(s,t)\cup\{G\in A_\phi\mid\text{$\exists a\in G$, $C(a)\subseteq E(s,t)$}\}$,
\item for all states $s,t\in W$, $E_{\phi}^+(s,t) = E_\phi(s,t) \cup \{G\in A_\phi \mid \exists n \geq 1, \exists s_0,\dots,s_n\in W, s_0=s \text{ and } s_n=t \text{ and } G\in \bigcap_{0\leq i< n}E_\phi(s_i,s_{i+1})\}$,
\end{itemize}
where without loss of generality we assume that $A_\phi \cap \ag = \emptyset$. We write $M_\phi^+$ for the quintuple $(W,A^\phi,E^+_\phi,C,\nu)$.
\qed
\end{definition}

\begin{proposition}
For any model $M$ and any formula $\phi$, $M^+_\phi$ is a model.\qed
\end{proposition}

\begin{lemma}\label{lem:trans-e}
Given formulas $\phi$ and $\chi$, a group $G$, a model $M$ and a state $s$ of $M$:
\begin{enumerate}
\item\label{it:trans-e1} $M,s\models\phi$ iff $M^+_{\chi},s\models\phi$;
\item\label{it:trans-e3} If ``$C_G$'' appear in $\chi$, then $M,s\models C_G\phi$ iff $M,t\models \phi$ for any state $t$ such that $G\in E^+_{\chi}(s,t)$.
\end{enumerate} 
\end{lemma}

\begin{proof}
\ref{it:trans-e1}. Notice that for any agent $a$, formula $\chi$ and states $s,t$, we have $C(a)\subseteq E(s,t)$ iff $C(a)\subseteq E_\chi(s,t)$ iff $C(a) \subseteq E^+_{\chi}(s,t)$. Thus it is easy to verify that $(M,s)$ , $(M_\chi,s)$ and $(M^+_{\chi},s)$ satisfy exactly the same formulas.

\ref{it:trans-e3}. We first observe the base case for $C_G\phi$:
$$\begin{array}[t]{rcll}
	M,s \models E_G \phi & \iff & \text{for any $a \in G$, $M,s\models K_a\phi$} & \\
	& \iff & \text{for any $a\in G$ and $t\in W$, $C(a)\subseteq E(s,t)$ implies $M,t \models \phi$} & \\
	& \iff & \text{for any $t\in W$ and $a\in G$, $C(a)\subseteq E(s,t)$ implies $M,t \models \phi$} & \\
	& \iff & \text{for any $t\in W$, if $C(a)\subseteq E(s,t)$ for some $a\in G$, then $M,t \models \phi$} & \\
	& \iff & \text{for any $t\in W$, $G\in E_\chi(s,t)$ implies $M,t \models \phi$} & \\
	& \iff & \text{$M,t\models \phi$ for any state $t$ such that $G\in E_\chi(s,t)$.} &\\
\text{and so}\quad	M,s \models C_G \phi & \iff & \text{$M,s\models E^k_G\phi$ for all $k\in \mbN^+$} & \\
	& \iff & \text{$M,t\models \phi$ for any state $t$ such that $G\in E^+_{\chi}(s,t)$} & (*)

\end{array}$$
where $(*)$ can be shown as follows: Suppose $M,s \not \models E^n_G \phi$ for some $n\in\mbN^+$, then by induction on $n$, we have $s_1, \dots, s_n \in W$ such that $M, s_n \not\models \phi$ and $G \in E_\chi(s,s_1)\cap\bigcap_{1\leq i< n}E_\chi(s_i,s_{i+1})$. Hence $M,s_n\not\models\phi$ and $G\in E^+_\chi(s,s_n)$.
Suppose $M,t\not\models \phi$ for some state $t$ such that $G\in E^+_\chi(s,t)$, w.l.o.g, assume that there exist $s_0,\dots,s_n\in W$ such that $s_0=s$, $s_n=t$, $G\in \bigcap_{0\leq i< n} E_\chi(s_i,s_{i+1})$ and $M,s_n \not\models \phi$. Thus using the above result $n$ times we have $M,s\not\models E^n_G \phi$.
\end{proof}

\begin{theorem}\label{thm:complexity}
The model checking problem for every proposed logic is in P.
\end{theorem}
\begin{proof}
The case for \l and \ls is given as Lemma~\ref{lem:mc-el}. For the extended logics, it suffices to provide a polynomial algorithm for the types of formulas $D_G\phi$, $E_G\phi$ (which is the base case for $C_G\phi$), $C_G\phi$ and $M_G\phi$. The details are given in Algorithm~\ref{alg:val-dcx}.

\begin{algorithm}
\caption{Function $Val(M,\phi)$ extended: cases with group knowledge operators}\label{alg:val-dcx}
\vspace{-1.2em}
\begin{multicols}{2}
\footnotesize
\begin{algorithmic}[1]
	\Initialize{$temVal \gets \emptyset$}
	\If{...} ... \Comment{This part keeps the same as in Algorithm~\ref{alg:val}}		
	\ElsIf{$\phi=D_G\psi$}
	\ForAll{$t\in W$}
	\Initialize{$n \gets \True$}
	\ForAll{$u \in W$}
	\If{$\bigcup_{a\in G}C(a)\subseteq E(t,u)$ \algand $u\not\in Val(M,\psi)$} {$n \gets \False$}
	\EndIf
	\EndFor
	\If{$n = \True$} $tmpVal \gets tmpVal \cup \{t\}$
	\EndIf
	\EndFor
	\State \Return $tmpVal$ \quad \Comment{Returns $\{t\in W\mid \forall u\in W: \bigcup_{a\in G}C(a)\subseteq E(t,u)\Rightarrow u\in Val(M,\psi)\}$}
	
	\ElsIf{$\phi = C_G\psi$}
	\ForAll{$t \in W$}
	\Initialize{$n \gets \True$}
	\ForAll{$u \in W$}
	\If{$G\in E^+_{\phi}(t,u)$ \algand $u \notin Val(M,\psi)$} {$n \gets \False$}
	\EndIf
	\EndFor
	\If{$n = \True$} $tmpVal \gets tmpVal \cup \{t\}$
	\EndIf
	\EndFor
	\State \Return $tmpVal$ \quad \Comment{Returns $\{t\in W\mid \forall u\in W: G\in E^+_{\phi}(t,u)\Rightarrow u\in Val(M,\psi)\}$}
	
	\ElsIf{$\phi = M_G\psi$}
	\ForAll{$t\in W$}
	\Initialize{$n \gets \True$}
	\ForAll{$u\in W$}
	\If{$\bigcap_{a\in G}C(a)\subseteq E(t,u)$ \algand $u\not\in Val(M,\psi)$} {$n \gets \False$}
	\EndIf
	\EndFor
	\If{$n = \True$} $tmpVal \gets tmpVal \cup \{t\}$
	\EndIf
	\EndFor
	\State \Return $tmpVal$ \quad \Comment{Returns $\{t\in W\mid \forall u\in W: \bigcap_{a\in G}C(a)\subseteq E(t,u)\Rightarrow u\in Val(M,\psi)\}$}
	\EndIf
\end{algorithmic}
\end{multicols}
\vspace{-.8em}
\end{algorithm}

As explained in the proof of Lemma \ref{lem:mc-el}, checking $C(a) \subseteq E(t,u)$ costs at most $|A|$ steps, here we furthermore need to the cost caused by the new modalities.

For $D_G$ and $M_G$, notice that the number of agents in any group $G$ that appears in $\phi$ is less than $|\phi|$, so checking $\bigcup_{a \in G} C(a)\subseteq E(t,u)$ and $\bigcap_{a \in G} C(a)\subseteq E(t,u)$ costs at most $|A| \cdot |\phi|$ steps. Thus for the logics extended with these modalities, the complexity for model checking would not go beyond P.

For $C_G$ (and its base case $E_G$), we need to ensure that there is a polynomial-time algorithm for computing $E_\phi(s,t)$ and $E^+_{\phi}(s,t)$ and checking whether $G$ is an element of them. By Definition~\ref{def:trans-e} and Lemma~\ref{lem:trans-e}, computing the set $A_\phi$ costs at most $|\phi|$ steps, since $C_G$ appears in $\phi$ at most $|\phi|$ times; and the size of $A_\phi$ is at most $|A|+ |\phi|$. To compute $E_\phi(s,t)$ for any given $s$ and $t$, it costs at most $|A| \cdot |\phi|$ steps to check whether there is $a\in G$ such that $C(a)\subseteq E(s,t)$. So the cost of computing the whole function $E_\phi$ can be finished in at most $|W|^2 \cdot |\phi| \cdot |A|$ steps.
Now we consider the computation of $E^+_{\phi}$. Assume that we have a string that describes $E_\phi$, then we check for all pairs $(s,t),(t,u)\in W^2$ whether there is a $G$ appearing in $\phi$ such that $G\in E_\phi(s,t)\cap E_\phi(t,u)$; if it is, we add $G$ as a member of $E_\phi(s,u)$. Keep doing this until $E_\phi$ does not change any more. Every round of checking takes at most $|\phi| \cdot (|A|+|\phi|)\cdot|W|^3$ steps, and it will be stable in at most $|\phi|\cdot|W|^2$ rounds. Then we obtain the function $E^+_{\phi}$ as we want. Every membership checking for $G\in E^+_{\phi}(t,u)$ is finished in polynomial steps. So the whole process is still in P.
\end{proof}

\subsection{The satisfiability/validity problem}
\label{sec:sat}

In this section, we explore the satisfiability and validity problems for the logics we have discussed. As articulated at the start of  Section~\ref{sec:complexity}, the two problems share the same complexity for every logic presented in this paper. We will show that all logics incorporating common knowledge exhibit EXPTIME-complete complexity for the satisfiability problem, while those devoid of common knowledge have PSPACE-completeness. The methodology for their proofs is depicted in Figure~\ref{fig:sat-withc} for logics with common knowledge and Figure~\ref{fig:sat-nonc} for logics without it.

The final results are presented in Theorem~\ref{thm:com-complex}, and the reduction methods used to obtain these results are provided in Lemmas~\ref{lem:red-withc}, \ref{lem:red-s52sc}, \ref{lem:red-sc2c}, \ref{lem:red-sdm2dm} and \ref{lem:red-dm2d}. We now proceed to validate these lemmas. For each lemma, we introduce a set of rewriting rules capable of transforming any satisfiable formula in the logic to be reduced into a satisfiable formula in the target logic. We present five such sets of rewriting rules. The verification of each lemma involves demonstrating the invariance and polynomial-time upper bound of the rewriting, a process that becomes lengthy when pursued in detail and generally follows the same pattern across all lemmas. For the sake of readability, we have relegated the detailed proof of rewriting invariance (namely, the proof of each proposition below) to a technical appendix at the end of the paper.

\begin{figure}
\footnotesize
\centering
$\xymatrix@C=2.7em{
&&*++o[F]{\lscdm}\ar@/^2.5pc/@{-->}[dd]^{\text{\begin{tabular}{@{}c@{}}
PTIME\\\tiny(Lemma~\ref{lem:red-withc})
\end{tabular}}}
&&*++o[F]{\lcdm}\ar@/_2.5pc/@{-->}[dd]
\\
&*++o[F]{\lscd}\ar@/^1.5pc/[ur]
&*++o[F]{\lscm}\ar[u]
&&*++o[F]{\lcm}\ar[u]
&*++o[F]{\lcd}\ar@/_1.5pc/[ul]
&*+[F]{\dfrac{\textrm{K}^C_n\ (n \geq 1)}{\text{\scriptsize EXPTIME complete}}}
\\
*+[F]{\dfrac{\textrm{S5}^C_n\ (n \geq 2)}{\text{\scriptsize EXPTIME complete}}}\ar@{-->}[rr]^{\text{PTIME}}_{\text{\tiny(Lemma~\ref{lem:red-s52sc})}}
&&*++o[F]{\lsc}\ar@/^1.5pc/[lu]\ar[u]\ar@{-->}[rr]^{\text{PTIME}}_{\text{\tiny(Lemma~\ref{lem:red-sc2c})}}
&& *++o[F]{\lc}\ar[u]\ar@/_1.5pc/[ur]\ar@/_1pc/[urr]
&& *+[F]{\dfrac{\textrm{K}^C_1}{\text{\scriptsize EXPTIME complete}}}\ar@{-->}@<2pt>[ll]_{\text{linear}}
}$
\caption{This figure illustrates the EXPTIME complete logics with common knowledge. A solid arrow from one logic to another represents the satisfiability/validity problem of the former logic as a subproblem of the satisfiability/validity problem for the latter. A dashed arrow labeled ``PTIME'' from one logic to another indicates that the satisfiability problem for the former logic can be reduced to the satisfiability problem for the latter using a polynomial-time algorithm. While ``PTIME'' indicates a polynomial-time reduction, ``linear'' signifies a linear-time reduction. These reductions are demonstrated in Lemmas~\ref{lem:red-withc}, \ref{lem:red-s52sc} and \ref{lem:red-sc2c}. The complexity results for \textrm{K}$_1^C$, \textrm{K}$_n^C$ and \textrm{S5}$_n^C$ are all taken from \cite[Section 3.5]{FHMV1995}.}\label{fig:sat-withc}
\end{figure}
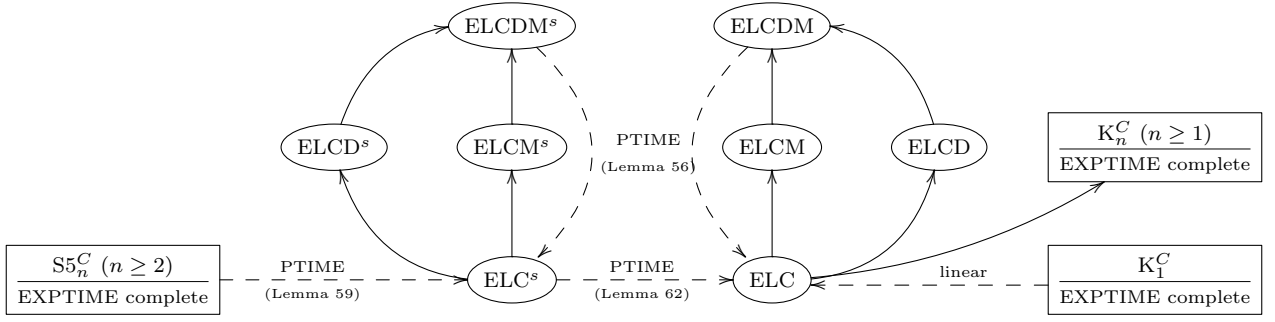

\begin{figure}
\footnotesize
\centering
$\xymatrix@C=3.3em{
&&*++o[F]{\lsdm}\ar@{-->}[rr]^{\text{PTIME}}_{\text{\tiny(Lemma~\ref{lem:red-sdm2dm})}}
&&*++o[F]{\ldm}\ar@/_1.5pc/@{-->}[dr]^{\text{\hspace{-6pt}\begin{tabular}{@{}c@{}}
PTIME\\[-1.2ex]\tiny(Lemma~\ref{lem:red-dm2d})
\end{tabular}}}
\\
&*++o[F]{\lsd}\ar@/^1.5pc/[ur]
&*++o[F]{\lsm}\ar[u]
&&*++o[F]{\lm}\ar[u]
&*++o[F]{\ld}\ar@/_1.5pc/[ul]\ar@<2pt>[r]
&*+[F]{\dfrac{\textrm{K}^D_n\ (n \geq 1)}{\text{\scriptsize PSPACE complete}}}\ar@<2pt>[l]
\\
*+[F]{\dfrac{\textrm{KB}_1}{\text{\scriptsize PSPACE complete}}}\ar[rr]
&&*++o[F]{\ls}\ar@/^1.6pc/[lu]\ar[u]
&& *++o[F]{\l}\ar@<0pt>[rr]\ar[u]\ar@/_1.5pc/[ur]
&& *+[F]{\dfrac{\textrm{K}_n\ (n \geq 1)}{\text{\scriptsize PSPACE complete}}}\ar@<3pt>[ll]
}$
\caption{This figure displays the PSPACE complete logics without common knowledge. It follows a similar convention to Figure~\ref{fig:sat-withc}. The complexity results of known satisfiability/validity problems are listed within the boxed frames. The results for \textrm{K}$_n$ and \textrm{K}$_n^D$ are derived from \cite[Section 3.5]{FHMV1995} (the subscript indicates the number of agents allowed in a logic), and the result for \textrm{KB}$_1$ is considered folklore (a clear reference can be found in [18], where it is named ``KB'' and refers to a 1992 manuscript for a proof).}\label{fig:sat-nonc}
\end{figure}
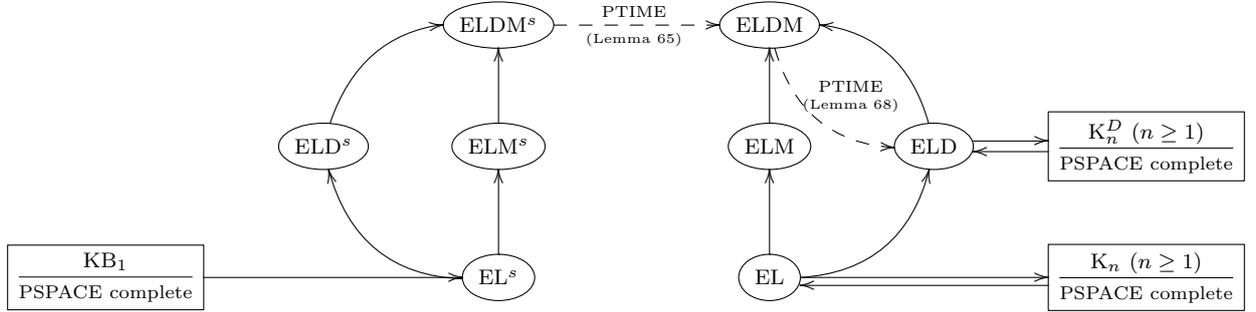

\paragraph{\bfseries\underline{Reductions from \lcdm to \lc, and from \lscdm to \lsc}.}

We will present a rule that rewrites any \langcdm-formula into an \langc-formula, preserving (s)-satisfiability in the process. To accomplish this, we will extend the set \ag of agents, adding a new agent for each operator $K_a$, $D_G$, and $M_G$ (for any agent $a$ and group $G$ in \ag). We denote these new agents as $f(K_a)$, $f(D_G)$, or $f(M_G)$, which means that $f$ is a function that maps every modal operator except $C_G$ (for the existing set \ag of agents) to a new agent. We use $\ag^+$ to represent this expanded set of agents, i.e., $\ag^+ = \ag \cup \{f(K_a), f(D_G), f(M_G) \mid a \in \ag \text{ and } \emptyset \neq G \subseteq \ag\}$. For any formula based on $\ag$, we denote $ag^+ (\phi)$ as the union of (1) the set of agents present in $\phi$ and (2) the set of new agents for the modal operators found in $\phi$. It is evident that $ag^+(\phi) \subseteq \ag^+$. 

\begin{definition}[rewriting rules]
\label{def:re1}
Consider an $\langcdm_\ag$-formula $\phi$. We define $\rho(\phi)$ as an $\langc_{\ag^+}$-formula constructed by sequentially applying the following steps:
\begin{enumerate}
\item Replace $\phi$ by $\phi \wedge \big( \big( \bigwedge_{\chi \in \mu(\phi)} \chi \big) \wedge C_{ag^+(\phi)} \big(\bigwedge_{\chi \in \mu(\phi)}\chi \big) \big)$. Here, $\mu(\phi)$ is the set of formulas of the following types:
(1) $M_G \psi \ra K_a\psi$, (2) $K_a\psi \ra D_G\psi$, (3) $M_H\psi \ra M_G\psi$, (4) $D_G \psi \ra D_H\psi$, (5) $M_I\psi \ra D_J\psi$.
In these, $a$ is any agent present in $\phi$, and $G, H, I, J$ are any groups of agents appearing in $\phi$, such that $a \in G \subseteq H$, $I \cap J \neq \emptyset$ and $\psi$ is a subformula of $\phi$;

\item For any agent $a \in \ag$, replace all instances of $K_a$ by $K_{f(K_a)}$;

\item For any group $G \subseteq \ag$, replace all instances of $D_G$ by $K_{f(D_G)}$, and $M_G$ by $K_{f(M_G)}$;

\item For any group $G \subseteq \ag$, replace all instances of $C_G$ by $C_{G'}$, where $G' = \{ f(K_a) \mid a \in G\}$. Note that the operator $C_{\ag^+(\phi)}$ generated in the first step is not replaced.
\end{enumerate}

We define $\rho'$ as a variant of $\rho$ that includes only of the last three steps. In other words, $\rho'(\phi)$ is obtained without initially replacing $\phi$ with the larger formula.
\qed
\end{definition}

\begin{propositionrep}[invariance of rewriting]\label{lem:sat-cdm2c}
Given an $\langcdm_\ag$-formula $\phi$, the following hold:
\begin{enumerate}
\item \label{it:sat-langc} Both $\rho(\phi)$ and $\rho'(\phi)$ are $\langc_{\ag^+}$-formulas;
\item \label{it:sat-cdm2c} The formula $\phi$ is satisfiable if and only if $\rho(\phi)$ is satisfiable;
\item \label{it:ssat-cdm2c} The formula $\phi$ is s-satisfiable if and only if $\rho(\phi)$ is s-satisfiable.
\end{enumerate}
\end{propositionrep}
\begin{inlineproof}
For the detailed proof, please refer to the appendix.
\end{inlineproof}
\begin{proof}
\ref{it:sat-langc}. \textit{Language Conversion}. In the first and second steps of the definition of $\rho$, new modalities are introduced, transforming the resulting formula into an element of $\langcdm_{\ag^+}$. The third step converts this into an $\langc_{\ag^+}$-formula, which remains unchanged in the fourth step.

\ref{it:sat-cdm2c}. Let us consider the direction from left to right. Begin by assuming that $\phi$ is satisfiable. By virtue of the soundness of \KCDM, $\phi$ is also \KCDM-consistent. Referring to the completeness theorem's proof (Section~\ref{sec:completeness5}), we find that $\phi$ is satisfied at a state $\langle \Delta^+\rangle$ in the standard model $\CM=(\CW,\CA,\CE,\CC,\CV)$ for \lcdm with respect to $cl(\phi)$. Here, $\Delta^+$ is a maximal $\KCDM$-consistent subset of $cl(\phi)$, and $\phi$ is an element of $\Delta^+$.

Consider the model $\CM_1=(\CW_1,\CA,\CE_1,\CC,\CV_1)$, where: 
\begin{itemize}
\item $\CW_1$ is the set of states $I$-reachable from $\langle \Delta^+ \rangle$.%
\footnote{A state $t$ is $I$-reachable from $s$ if there exists a list of states $s_0, \dots, s_n$ ($n \in \mbN$) such that $s_0=s$, $s_n=t$, and $\bigcap_{ a \in I }\CC(a) \subseteq \CE(s_{i-1}, s_i)$ for all $1 \leq i \leq n$.}
Here, $I$ denotes the set of agents present in $\phi$.

\item $\CE_1$ and $\CV_1$ are restrictions of $\CE$ and $\CV$ on $\CW_1$, respectively.
\end{itemize}
Since any group $G$ in $\phi$ is a subset of $I$, we can inductively show that for every subformula $\psi$ of $\phi$ and each state $s \in \CW_1$, $\CM_1,s\models\psi$ if and only if $\CM,s\models\psi$. In particular, $\CM_1,\langle \Delta^+\rangle\models\phi$, and for every state $s\in\CW_1$ and formula $\chi\in\mu(\phi)$, $\CM_1,s\models\chi$.

Next, consider the model $\CM_2=(\CW_1,\CA,\CE_1,\CC_2,\CV_1)$, where $\CC_2:\ag^+\to\CA$ is defined as follows:
$$\CC_2(a) = \left\{\begin{array}{ll}
\CC(a), & \text{if $a\in\ag$},\\
\CC(b), & \text{if $a$ is $f(K_b)$ for some $b\in\ag$},\\
\bigcup_{a\in G} \CC(a), & \text{if $a$ is $f(D_G)$ for some group $G$ on $\ag$},\\
\bigcap_{a\in G} \CC(a), & \text{if $a$ is $f(M_G)$ for some group $G$ on $\ag$}.
\end{array}\right.$$
For any $\langcdm_\ag$-formula $\psi$ and any state $s\in\CW_1$, we can confirm that (1) $\CM_1,s\models \psi \Longleftrightarrow \CM_2,s\models \psi$, and (2) $\CM_1,s\models \psi \Longleftrightarrow \CM_2,s\models \rho'(\psi)$. From the first condition, we deduce that $\CM_2,\langle\Delta^+\rangle\models \phi \wedge \big(\big(\bigwedge_{\chi \in \mu(\phi)} \chi \big) \wedge C_{ag^+(\phi)} \big(\bigwedge_{\chi \in \mu(\phi)}\chi \big) \big)$. From the second condition, we infer that $\CM_2,\langle\Delta^+\rangle\models\rho(\phi)$. Hence, we conclude that $\rho(\phi)$ is satisfiable.

Consider the scenario from right to left. Let us suppose that $\rho(\phi)$ is satisfiable, hence it is \KC-consistent. The proof of the completeness theorem (Theorem \ref{thm:completeness6}) demonstrates that it is satisfied at a state $\langle\Delta^+\rangle$ in the standard model $\CM=(\CW,\CA,\CE,\CC,\CV)$ for \lc with respect to $cl(\rho(\phi))$. In this context, $\Delta^+$ is a maximal $\KC$-consistent subset of $cl(\rho(\phi))$ and $\rho(\phi) \in \Delta^+$. We then consider the model $\CM_1=(\CW_1,\CA,\CE_1,\CC,\CV_1)$ where (1) $\CW_1$ is the set of states reachable from $\langle\Delta^+\rangle$, 
(2) $\CE_1$ is the restriction of $\CE$ on $\CW_1$, and (3) $\CV_1$ is the restriction of $\CV$ on $\CW_1$.
Therefore, for any $s\in\CW_1$ and any $\langc_{\ag^+}$-formula $\psi$, $\CM_1,s\models\psi \Longleftrightarrow \CM,s\models\psi$. Specifically,  $\CM_1,\langle\Delta^+\rangle\models\rho(\phi)$.

Next we consider the model $\CM_2=(\CW_1,\CA,\CE_2,\CC,\CV_1)$, where $\CE_2:(\CW_1)^2 \to \CA$ is given by the following:
$$\CE_2(s,t)  = \left\{\begin{array}{ll}
	\CC(a), & \text{if $t$ is $s$ extended with $\langle (\{f(K_a)\},d), \Psi \rangle$ and $\{\chi\mid K_{f(K_a)}\chi \in tail(s)\}  \subseteq \Psi$},\\
	\bigcup_{a\in G} \CC(a), & \text{if $t$ is $s$ extended with $\langle (\{f(D_G)\},d), \Psi \rangle$ and $\{\chi\mid K_{f(D_G)}\chi \in tail(s)\}  \subseteq \Psi$},\\
	\bigcap_{a\in G} \CC(a), & \text{if $t$ is $s$ extended with $\langle (\{f(M_G)\},d), \Psi \rangle$ and $\{\chi\mid K_{f(M_G)}\chi \in tail(s)\}  \subseteq \Psi$},\\
	\emptyset, & \text{otherwise}.
\end{array}\right.$$
Note that for any $\langcdm_{\ag}$-formula $\psi$ and any $s\in\CW_1$, if $\psi$ is a subformula of $\phi$, then $\CM_1,s \models \rho'(\psi) \Longleftrightarrow \CM_2,s\models \psi$. This can be demonstrated by induction on $\psi$. The boolean cases are straightforward. Here we only prove the case for $\psi = C_G \chi$. Other modalities can be shown in a similar manner:

Assume $\CM_1,s\not\models \rho'(C_G\chi)$, which implies $\CM_1,s\not\models C_{G'}\rho'(\chi)$. This signifies the existence of $a_1,\dots,a_n$ ($n\in\mbN^+$) such that $\CM_1,s\not\models K_{f(K_{a_1})}\cdots K_{f(K_{a_n})}\rho'(\chi)$. Consequently, there must be $s_0,\dots,s_n$ ($n\in\mbN$) fulfilling the conditions: $s_0=s$, $\CM_1,s_n\not\models \rho'(\chi)$, and for any $1\leq i\leq n$, $\CC(f(K_{a_i}))\subseteq \CE_1(s_{i-1},s_i)$. For each $1\leq i\leq n$, $s_i$ is the extension of $s_{i-1}$ with $\langle ({f(K_{a_i})},d), \Psi_i \rangle$, where $\Psi_i$ is a maximal \KC-consistent subset of $cl(\rho(\phi))$. Moreover, $\{\chi\mid K_{f(K_{a_i})}\chi \in tail(s) \} \subseteq \Psi_i$. For each $1\leq i\leq n$, it holds that $\CC(a_i)\subseteq \CE_2(s_{i-1},s_i)$. The induction hypothesis implies that $\CM_2,s_n\not\models \chi$, leading to the conclusion $\CM_2,s_n\not\models C_G\chi$.

Now, let's assume $\CM_2,s\not\models C_G\chi$. This implies the existence of $a_1,\dots,a_n$ ($n\in\mbN^+$) such that $\CM_2,s\not\models K_{a_1}\cdots K_{a_n}\chi$. Consequently, there must be $s_0,\dots,s_n$ ($n\in\mbN$) fulfilling the conditions: $s_0=s$, $\CM_2,s_n\not\models \chi$, and for any $1\leq i\leq n$, $\CC(a_i)\subseteq \CE_2(s_{i-1},s_i)$. For each $1\leq i\leq n$, either: (1) $s_i$ is the extension of $s_{i-1}$ with $\langle ({f(K_{a_i})},d), \Psi_i \rangle$ and $\{\chi\mid K_{f(K_{a_i})}\chi \in tail(s)\} \subseteq \Psi_i$, (2) $s_i$ is the extension of $s_{i-1}$ with $\langle ({f(D_{G_i})},d), \Psi_i \rangle$ and $\{\chi\mid K_{f(D_{G_i})}\chi \in tail(s)\} \subseteq \Psi_i$, or (3) $s_i$ is the extension of $s_{i-1}$ with $\langle (f(M_{\{a_i\}}), d), \Psi_i \rangle$ and $\{ \chi \mid K_{f (M_{\{a_i\}})} \chi \in tail(s)\} \subseteq \Psi_i$. Here, $\Psi_i$ is a maximal \KC-consistent subset of $cl(\rho(\phi))$ and $G_i$ is a group of agents in $\ag$ such that $a_i\in G_i$. For any $\eta$, if $K_{f(K_{a_i})}\eta\in cl(\rho(\phi))$, then either (i) $\eta \in \{ \rho'(\psi) \mid \text{$\psi$ is a subformula of $\phi$} \}$, (ii) $\eta=\rho'\big(\bigwedge_{\chi \in \mu(\phi)} \chi \big)$, or (iii) $\eta=\rho'\big(C_{ag^+(\phi)}\bigwedge_{\chi \in \mu(\phi)} \chi \big)$. Therefore, for any $\eta$, if $K_{f(K_{a_i})}\eta\in tail(s)$, it implies $K_{f(D_{G_i})}, K_{f(M_{\{a_i\}})}\eta\in tail(s)$. Thus, in all cases (1)--(3) we have $\{\chi\mid K_{f(K_{a_i})}\chi \in tail(s)\} \subseteq \Psi_i$. 
Let us define a sequence of canonical paths $s'_0,\dots,s'_n$ ($n\in\mbN$) such that $s_0=s$, and for each $1 \leq i \leq n$, $s'_i$ is the extension of $s'_{i-1}$ with $\langle (\{f(K_{a_i})\},d), \Psi_i \rangle$. For each $1\leq i\leq n$, we have $\CC(f(K_{a_i}))\subseteq \CE_1(s'_{i-1},s'_i)$. By applying the truth lemma, we conclude $\CM_2,s'_n\not\models \chi$. The induction hypothesis implies $\CM_1,s'_n\not\models \rho'(\chi)$, leading to the conclusion $\CM_1,s\not\models C_{G'}(\rho'(\chi))$, i.e., $\CM_1,s\not\models \rho'(C_G\chi)$.

Thus, $\CM_1,\langle\Delta^+\rangle\models \rho'(\phi)$ if and only if $\CM_2,\langle\Delta^+\rangle\models \phi$. Given that $\CM_1,\langle\Delta^+\rangle\models \rho(\phi)$, it follows that $\CM_2,\langle\Delta^+\rangle\models \phi$, establishing that $\phi$ is satisfiable.

Regarding \ref{it:ssat-cdm2c}, the argument parallels that of \ref{it:sat-cdm2c}, with the distinction that we begin with standard models for \lscdm that function as similarity models.
\end{proof}

\begin{lemma}\label{lem:red-withc}
The following hold:
\begin{enumerate}
\item \label{it:red-cdm2c} The satisfiability problem of \lcdm is polynomial time reducible to that of \lc;
\item \label{it:red-scdm2sc} The satisfiability problem of \lscdm is polynomial time reducible to that of \lsc.
\end{enumerate}
\end{lemma}
\begin{proof}
Consider an $\langcdm_\ag$-formula $\phi$. As per Lemma \ref{lem:sat-cdm2c}, the $\langc_{\ag^+}$-formula $\rho(\phi)$ is constructed such that $\phi$ is (s-)satisfiable if and only if $\rho(\phi)$ is (s)-satisfiable. Therefore, it suffices to show that the $\rho$ procedure is polynomial in the size of $\phi$. Assume that the size of $\phi$ is $k$. The execution of the first step in computing $\rho(\phi)$ (as per Definition~\ref{def:re1}) is polynomial in $k$, as it merely involves listing the formulas in $\mu(\phi)$ and binding them with $\wedge$ and $C_{\ag^+(\phi)}$. The size of $\mu(\phi)$ is polynomial, given that: (1) the number of subformulas of $\phi$ is at most $k$, (2) the number of modal operators present in $\phi$ is at most $k$, and (3) the size of any group appearing in $\phi$ is at most $k$.
\end{proof}

\paragraph{\bfseries \underline{Reduction from $S5^C_n$ to \lsc}.}

Next, we propose a rule for converting any \langc-formula that is satisfiable in the logic $S5^C_n$ ($n \geq 1$) into an \langc-formula that is s-satisfiable.

\begin{definition}\label{def:rewrite-t}
Given an $\langc$-formula $\phi$, we define $\rho^t(\phi)$ as $\phi \wedge \big(\bigwedge_{\chi \in \mu^t(\phi)} \chi \big) \wedge C_{\ag} \bigwedge_{\chi \in \mu^t(\phi)}\chi$. Here, $\mu^t(\phi)$ includes the following types of formulas: (1) $K_a \psi \ra K_a K_a\psi$, (2) $K_a \psi \ra \psi$, and (3) $\neg K_a\bot$. In these, $a$ is any agent present in $\phi$ and $\psi$ is a subformula of $\phi$.
\qed
\end{definition}
It is evident from the definition above that if $\phi$ is an $\langc$-formula, so is $\rho^t (\phi)$.

\begin{propositionrep}\label{lem:sat-s52sc}
Given an \langc-formula $\phi$, the following hold:
\begin{enumerate}
\item \label{it:sat-s52sc-r} $\phi$ is satisfied at a state in an S5 relational model iff $\rho^t(\phi)$ is satisfied at a state in a symmetric relational model;
\item \label{it:sat-s52sc} $\phi$ is satisfiable in $S5^C_n$ (for $n \geq 1$) if and only if $\rho^t(\phi)$ is s-satisfiable.
\end{enumerate}
\end{propositionrep}
\begin{inlineproof}
Detailed proof can be found in the appendix.
\end{inlineproof}
\begin{proof}
\ref{it:sat-s52sc-r}. Suppose $\phi$ is satisfied at a state $s$ in an S5 relational model $N=(W,R,V)$. It can be readily confirmed that $N,s\Vdash\rho^t(\phi)$, indicating that $\rho^t(\phi)$ is satisfied at a state in a symmetric relational model.

Now suppose $\rho^t(\phi)$ is satisfied at a state $s$ in a symmetric relational model $N=(W,R,V)$. Without loss of generality, we can assume that all states in $W$ are reachable from $s$. Let us define a new model $N'=(W,R',V)$ where $R'$ is the reflexive and transitive closure of $R$. We can then demonstrate that for any subformula $\psi$ of $\phi$ and $t\in W$: $N,t\Vdash\psi$ if and only if $N',t\Vdash\psi$. The proof is done by induction on $\psi$, and below we show the cases with modalities.

For the case $\psi = K_a\chi$, suppose $N,t\not\Vdash K_a\chi$, then there is some $t'$ such that $R_a(t,t')$ and $N,t'\not\Vdash\chi$. Obviously we have $R'_a(t,t')$ and by induction hypothesis we have $N',t'\not\Vdash\chi$, and hence $N',t\not\Vdash K_a\chi$.
Suppose $N',t\not\Vdash K_a\chi$, then there exists $t'\in W$ such that $R'_a(t,t')$ and $N',t'\not\Vdash\chi$. Then there exists $t_0,\cdots,t_n$, $n\in\mbN$ such that $t_0=t$, $t_n=t'$ and $R_a(t_{i-1},t_i)$ for any $1\leq i\leq n$. By induction hypothesis we have $N,t'\not\Vdash\chi$. If $n=0$, then $t'=t$, thus we have $N,t\not\Vdash\chi$, and since $N,t\Vdash K_a\chi\ra \chi$, we have $N,t\not\Vdash K_a\chi$. If $n=1$, then we have $N,t\not\Vdash K_a\chi$ immediately. If $n>1$, then we have $N,t_{n-1}\not\Vdash K_a\chi$ and $N,t_{n-2}\not\Vdash K_aK_a\chi$. Since $N,t_{n-2}\Vdash K_a\chi\ra K_aK_a\chi$ we have $N,t_{n-2}\not\Vdash K_a\chi$. By proving step by step on $n\geq i\geq 0$, we can finally show that $N,t_0\not\Vdash K_a\chi$, i.e., $N,t\not\Vdash K_a\chi$.

For $\psi = C_G \chi$, since $N,s\Vdash \rho^t(\phi)$ and every state of $N$ is reachable from $s$, $N$ is symmetric and serial. Therefore, it can be shown that the transitive closure of $\bigcup_{a\in G}R'(a)$ is the same as the transitive closure of $\bigcup_{a\in G}R(a)$, and hence the case when $\psi = C_G\chi$ is also proven.

Finally, since $N,s\Vdash\rho^t(\phi)$, we have $N,s\Vdash\phi$. Thus $N',s\Vdash\phi$, and so $\phi$ is satisfied at a state in an S5 relational model.

\ref{it:sat-s52sc}. The satisfiability of $\phi$ in $S5^C_n$ is equivalent to $\rho^t(\phi)$ being satisfied at a state in a symmetric relational model (by \ref{it:sat-s52sc-r}), which is in turn equivalent to $\rho^t(\phi)$ being satisfied at a state in a similarity model (by Lemmas~\ref{lem:trans-r}--\ref{lem:tr-sim-model}), and thus equivalent to $\rho^t(\phi)$ being s-satisfiable.
\end{proof}

\begin{lemma}\label{lem:red-s52sc}
The satisfiability problem of $S5^C_n$ is polynomial time reducible to that of \lsc.
\end{lemma}
\begin{proof}
Choose a set \ag of agents such that $|\ag|=n$. Then, the function $\rho^t$ can reduce the satisfiability problem of $S5^C_n$ to that of \lsc in polynomial time.
\end{proof}

\paragraph{\bfseries\underline{Reduction from \lsc to \lc}.}

\begin{definition}\label{def:rewrite-s}
Given an $\langc$-formula $\phi$, we define the rewriting rule $\rho^s$ such that $\rho^s(\phi)$ is as follows:
$$\textstyle \rho^s(\phi)=\phi \wedge \big(\bigwedge_{a\in\ag,\ \psi\in Sub(\phi)} (\neg K_a \neg K_a\psi \ra \psi) \big) \wedge C_{\ag} \bigwedge_{a\in\ag,\ \psi\in Sub(\phi)}(\neg K_a \neg K_a\psi \ra \psi),$$
where $Sub(\phi)$ stands for the set of all subformulas of $\phi$.
\qed
\end{definition}
It is evident that if $\phi$ is an \langc-formula, then $\rho^s(\phi)$ is also an \langc-formula.

\begin{propositionrep}\label{lem:sat-sc2c}
Given an $\langc$-formula $\phi$, the following statements are true:
\begin{enumerate}
\item \label{it:sat-sc2c-r} $\phi$ is satisfied at a state in a symmetric relational model iff $\rho^s(\phi)$ is satisfied at a state in a relational model;
\item \label{it:sat-sc2c} $\phi$ is s-satisfiable if and only if $\rho^s(\phi)$ is satisfiable.
\end{enumerate}
\end{propositionrep}
\begin{inlineproof}
The detailed proof can be found in the appendix.
\end{inlineproof}
\begin{proof}
\ref{it:sat-sc2c-r}. Let us first assume that $\phi$ is satisfied at a state $s$ in a symmetric relational model $N=(W,R,V)$. Given that $N$ is a symmetric relational model, we can easily show that $N,s\Vdash\rho^s(\phi)$. This implies that $\rho^s(\phi)$ is satisfied at a state in a relational model.

Next, let us assume that $\rho^s(\phi)$ is satisfied at a state $s$ in a relational model $N=(W,R,V)$. Without loss of generality, we can assume that all states in $W$ are reachable from $s$. We define a new model $N'=(W,R',V)$ where $R'$ is the symmetric closure of $R$. We claim that for any subformula $\psi$ of $\phi$ and $t \in W$: $N,t\Vdash\psi \Longleftrightarrow N',t\Vdash\psi$. This claim is proven by induction on $\psi$, where we only show the cases involving the common knowledge operators.

Consider the case $\psi=C_G\chi$ for some $G\subseteq\ag$.
\begin{itemize}
\item Suppose $N,t\not\Vdash C_G\chi$, then there exist $t_0,\dots,t_n\in W$ and $a_1,\dots,a_n$ ($n\in\mbN^+$) such that $t_0=t$, $N,t_n\not\Vdash\chi$ and for any $1\leq i\leq n$, $R_{a_i}(t_{i-1},t_i)$ and $a_i\in G$. Hence for any $1\leq i\leq n$, $R'_{a_i}(t_{i-1},t_i)$ and $N',t_n\not\Vdash\chi$ by the induction hypothesis. Thus $N',t\not\Vdash C_G\chi$.
\item Suppose $N',t\not\Vdash C_G\chi$, there exist $t_0,\dots,t_n\in W$ and $a_1,\dots,a_n$ ($n\in\mbN^+$) such that $t_0=t$, $N',t_n\not\Vdash\chi$ and for any $1\leq i\leq n$, $R'_{a_i}(t_{i-1},t_i)$ and $a_i\in G$. By the induction hypothesis we have $N,t_n\not\Vdash\chi$. Notice that for any $1\leq i\leq n$, either $R_{a_i}(t_{i-1},t_i)$ or $R_{a_i}(t_i,t_{i-1})$. In particular, either $R_{a_n}(t_{n-1},t_n)$ or $R_{a_n}(t_n,t_{n-1})$. For the former we get $N,t_{n-1}\not\Vdash K_{a_n}\chi$ immediately. When $R_{a_n}(t_n,t_{n-1})$, since $N,t_n\Vdash\neg K_{a_n}\neg K_{a_n}\chi\ra\chi$, we have $N,t_n\not\Vdash\neg K_{a_n}\neg K_{a_n}\chi$, and then $N,t_{n-1}\not\Vdash K_{a_n}\chi$. Thus we have $N,t_{n-1}\not\Vdash C_G\chi$. Since $C_G\chi$ is a subformula of $\phi$, when $n\geq 2$ we can similarly get $N,t_{n-2}\not\Vdash K_{a_{n_1}}C_G\chi$. And then we get $N,t_{n-2}\not\Vdash C_G\chi$. So by induction on $1\leq i\leq n$ we can prove that $N,t_{n-i}\not\Vdash C_G\chi$. In particular, $N,t\not\Vdash C_G\phi$.
\end{itemize}
Having proven the above, we know that since $N,s\Vdash\rho^s(\phi)$, it follows that $N,s\Vdash\phi$. By the claim, we have $N',s\Vdash\phi$, where $N'$ is a symmetric relational model. Thus, $\phi$ is satisfied at a state in a symmetric relational model.

\ref{it:sat-sc2c}. The statement ``$\phi$ is s-satisfiable'' is equivalent to``$\phi$ is satisfied at a state in a symmetric relational model'' (by Lemmas~\ref{lem:trans-r}--\ref{lem:tr-sim-model}), which is equivalent to ``$\rho^s(\phi)$ is satisfied at a state in a relational model'' (by \ref{it:sat-sc2c-r}), which in turn is equivalent to ``$\rho^s(\phi)$ is satisfiable'' (by Lemmas~\ref{lem:trans-r}--\ref{lem:trans-w}).
\end{proof}

\begin{lemma}\label{lem:red-sc2c}
The satisfiability problem of \lsc is polynomial time reducible to that of \lc.
\qed
\end{lemma}

\paragraph{\bfseries\underline{Reduction from \lsdm to \ldm}.}
\begin{definition}\label{def:rewrite-m}
Given an $\langdm$-formula $\phi$, we define $\rho^m(\phi)$ as $\phi \wedge \bigwedge_{\chi \in \mu^s(\phi),\ 0\leq i\leq |\phi|} M_\ag^i \chi$.
Here, $\mu^m(\phi)$ comprises the following types of formulas:
(1) $\neg K_a \neg K_a\psi \ra \psi$, (2) $\neg D_G \neg D_G\psi \ra \psi$, and (3) $\neg M_G \neg M_G\psi \ra \psi$, where ``$a$'' denotes any agent appearing in $\phi$, ``$G$'' refers to any group of agents present in $\phi$, and $\psi$ is a subformula of $\phi$.
\qed
\end{definition}

It is evident from the definition above that if $\phi$ is an $\langdm$-formula, then $\rho^m(\phi)$ is also an $\langdm$-formula.

\begin{propositionrep}\label{lem:sat-sdm2dm}
Given an $\langdm$-formula $\phi$, $\phi$ is s-satisfiable if and only if $\rho^m(\phi)$ is satisfiable.
\end{propositionrep}
\begin{inlineproof}
Detailed proof can be found in the appendix.
\end{inlineproof}
\begin{proof}
Suppose $\phi$ is satisfied at a state $s$ in a similarity model $M=(W,\ab,E,C,\nu)$. Given that $M$ is a similarity model, it can be easily verified that $M,s\models\rho^m(\phi)$. This implies that $\rho^m(\phi)$ is satisfied at a state in a model.

Next, let us assume $\rho^m(\phi)$ is satisfied at a state $s$ in a model $M=(W,\ab,E,C,\nu)$. Without loss of generality, we can assume that all states in $W$ are $\ag$-reachable from $s$.
We define a new model $M'=(W,\ab,E',C,\nu)$ where for any $t,t'\in W$, $E'(t,t')=E(t,t')\cup E(t',t)$. $M'$ is symmetric. We assert that for any subformula $\psi$ of $\phi$ and $t\in W$ which is reachable from $s$ in at most $|\phi|-|\psi|$ steps: $W,t\models\psi \Longleftrightarrow W',t\models\psi$. We prove this claim by induction on $\psi$.

Consider the case $\psi=K_a\chi$ for some $a\in\ag$. The arguments for this case involve some reasoning based on the relationships between states and the properties of the models, and the result can be shown using induction.We direct the reader to the preceding proofs for a deeper understanding of such proof techniques. The cases for $\psi=D_G\chi$ and $\psi=M_G\chi$ follow a similar argumentative pattern and are thus omitted for brevity.

Having proven the above, and given that $M,s\models\rho^m(\phi)$, we can conclude that $M,s\models\phi$. Thus, we have $M',s\models\phi$, where $M'$ is a symmetric model. Therefore, $\phi$ is satisfied at a state in a symmetric model. By Lemma \ref{lem:tr-sim-model}, $\phi$ is also satisfied at a state in a similarity model, hence $\phi$ is s-satisfiable.
\end{proof}

\begin{lemma}\label{lem:red-sdm2dm}
The satisfiability problem of \lsdm is polynomial time reducible to that of \ldm.
\qed
\end{lemma}

\paragraph{\bfseries\underline{Reduction from \ldm to \ld}.}
\begin{definition}\label{def:rewrite-d}
Given an $\langdm_\ag$-formula $\phi$, we define $\tau(\phi)$ as an $\langd_{\ag^+\cup\{o\}}$-formula obtained by sequentially performing the following steps, where ``$o$'' is a new agent not present in $\ag^+$:
\begin{enumerate}
\item Replace $\phi$ by $\phi \wedge \neg K_o^{|\phi|}\bot \wedge \bigwedge_{\chi \in \mu(\phi),\ 0\leq i\leq |\phi|} K_o^i \chi$. Here, $\mu(\phi)$ is the collection of the following types of formulas:
(1) $M_G \psi \ra K_a\psi$, (2) $K_a\psi \ra D_G\psi$, (3) $M_H\psi \ra M_G\psi$, (4) $D_G \psi \ra D_H\psi$, and (5) $M_I\psi \ra D_J\psi$.
In these formulas, ``$a$'' is any agent appearing in $\phi$, $G, H, I, J$ are any groups of agents present in $\phi$ such that $a \in G \subseteq H$ and $I \cap J \neq \emptyset$, and $\psi$ is a subformula of $\phi$;

\item For any group $G$, replace all instances of $D_G$ by $D_{\{o,f(D_G)\}}$, and $M_G$ by $D_{\{o,f(M_G)\}}$;

\item For any agent $a$, replace all instances of $K_a$ by $D_{\{o,f(K_a)\}}$.
\end{enumerate}
The rewriting rule $\tau'$ operates similarly to $\tau$, except it omits the first step.
\qed
\end{definition}

\begin{propositionrep}\label{lem:sat-dm2d}
Given an $\langdm_\ag$-formula $\phi$, the following hold:
\begin{enumerate}
\item \label{it:sat-langdm} $\tau(\phi)$ and $\tau'(\phi)$ are $\langd_{\ag^+\cup\{o\}}$-formulas;
\item \label{it:sat-dm2d} $\phi$ is satisfiable if and only if $\tau(\phi)$ is satisfiable.
\end{enumerate}
\end{propositionrep}
\begin{inlineproof}
Detailed proof can be found in the appendix.
\end{inlineproof}
\begin{proof}
\ref{it:sat-langdm}. As per the definition of rewriting rules, all the mutual modalities are replaced by distributed modalities with agents in $\ag^+\cup\{o\}$. Thus, $\tau(\phi)$ and $\tau'(\phi)$ are formulas of $\langd_{\ag^+\cup\{o\}}$.

\ref{it:sat-dm2d}. 
\emph{From left to right}. 
Assume $\phi$ is satisfiable, and thus is \KDM-consistent. By the proof of the completeness theorem (Theorem~\ref{thm:completeness5}), it is satisfied at a state $\langle\Delta^+\rangle$ in the standard model $\CM=(\CW,\CA,\CE,\CC,\CV)$ for \ldm, where $\Delta^+$ is a maximal $\KDM$-consistent set of formulas and $\phi\in\Delta^+$.

Next, consider a model $\CM_1=(\CW_1,\CA,\CE_1,\CC,\CV_1)$ where $\CW_1$ is the set of states $I$-reachable from $\langle\Delta^+\rangle$ where $I$ is the set of all agents present in $\phi$, and $\CE_1$ and $\CV_1$ are the restrictions of $\CE$ and $\CV$ on $\CW_1$, respectively. Since any group $G$ present in $\phi$ is a subset of $I$, by induction for any subformula $\psi$ of $\phi$ we can easily verify that for any $s\in\CW_1$: $\CM_1,s\models\psi \Longleftrightarrow \CM,s\models\psi$. In particular, $\CM_1,\langle\Delta^+\rangle\models\phi$. Moreover, for any $s\in\CW_1$ and $\psi\in\mu(\phi)$, we have $\CM_1,s\models\psi$.

We then consider another model $\CM_2=(\CW_1,\CA,\CE_1,\CC_2,\CV_1)$, where $\CC_2:(\ag^+\cup\{o\}) \to \CA$ is defined as follows:
$$\CC_2(a)  = \left\{\begin{array}{ll}
	\CC(a), & \text{if $a\in\ag$},\\
	\CC(b), & \text{if $a$ is $f(K_b)$ for some $b\in\ag$},\\
	\bigcup_{a\in G} \CC(a), & \text{if $a$ is $f(D_G)$ for some group $G$ on $\ag$},\\
	\bigcap_{a\in G} \CC(a), & \text{if $a$ is $f(M_G)$ for some group $G$ on $\ag$},\\
	\emptyset, & \text{if $a$ is $o$}.
\end{array}\right.$$
This model serves to convert the initial model into a form that allows us to verify the formula $\phi$.

It can be verified that for any $\langdm_\ag$-formula $\psi$ and any $s\in\CW_1$, $\CM_1,s \models \psi \Longleftrightarrow \CM_2,s\models \psi \Longleftrightarrow \CM_2,s\models \tau'(\psi)$. Since $\CC_2(o)=\emptyset$, we have $\CM_2,\langle\Delta^+\rangle\models \phi \wedge \neg K_o^{|\phi|}\bot \wedge \bigwedge_{\chi \in \mu(\phi), 0\leq i\leq |\phi|} K_o^i \chi$, and hence $\CM_2,\langle\Delta^+\rangle\models\tau(\phi)$.

\emph{From right to left.} Suppose $\tau(\phi)$ is satisfiable, thus it is \KD-consistent, then by the proof of the completeness theorem (Theorem \ref{thm:completeness5}), it is satisfied at a state $\langle\Delta^+\rangle$ in the standard model $\CM=(\CW,\CA,\CE,\CC,\CV)$ for \ld, where $\Delta^+$ is a maximal $\KC$-consistent set of formulas and $\tau(\phi) \in \Delta^+$.

Again, we consider a model $\CM_1=(\CW_1,\CA,\CE_1,\CC,\CV_1)$, where $\CW_1$ is the set of states which are $o$-reachable from $\langle\Delta^+\rangle$, and $\CE_1$ and $\CV_1$ are the restrictions of $\CE$ and $\CV$ on $\CW_1$, respectively. For any subformula $\psi$ of $\tau'(\phi)$, and any $s\in\CW_1$, $\CM_1,s\models\psi \Longleftrightarrow \CM,s\models\psi$. In particular, $\CM_1,\langle\Delta^+\rangle\models\tau'(\phi)$. Similarly we have $\CM_1,\langle\Delta^+\rangle\models\neg K_o^{|\phi|}\bot$. Moreover, since $\CM_1,s\models\chi$ for any $\chi\in\mu(\phi)$ and any $s\in\CW_1$ which is $o$-reachable from $\langle\Delta^+\rangle$ in at most $|\phi|$ steps, we have $\CM_1,\langle\Delta^+\rangle\models\tau(\phi)$.

Next, we consider the a model $\CM_2=(\CW_1,\CA,\CE_2,\CC,\CV_1)$, where $\CE_2:(\CW_1)^2 \to \CA$ is given as follows:
		$$\CE_2(s,t)  = \left\{\begin{array}{ll}
			\CC(a), & \text{if $t$ is $s$ extended with $\langle (\{o,f(K_a)\},d), \Psi \rangle$ and $\{\chi\mid D_{\{o,f(K_a)\}}\chi \in tail(s)\}  \subseteq \Psi$},\\
			\bigcup_{a\in G} \CC(a), & \text{if $t$ is $s$ extended with $\langle (\{o,f(D_G)\},d), \Psi \rangle$ and $\{\chi\mid D_{\{o,f(D_G)\}}\chi \in tail(s)\}  \subseteq \Psi$},\\
			\bigcap_{a\in G} \CC(a), & \text{if $t$ is $s$ extended with $\langle (\{o,f(M_G)\},d), \Psi \rangle$ and $\{\chi\mid D_{\{o,f(M_G)\}}\chi \in tail(s)\}  \subseteq \Psi$},\\
			\emptyset, & \text{otherwise}.
		\end{array}\right.$$
Similar to the previous case, we can show that for any $\langdm_{\ag}$-formula $\psi$ and any $s\in\CW_1$, if $\psi$ is a subformula of $\phi$ and $s$ is reachable from $\langle \Delta^+ \rangle$ in at most $|\phi| - |\psi|$ steps, then $\CM_1,s\models \tau'(\psi) \Longleftrightarrow \CM_2,s\models \psi$. Hence, if $\CM_1,s\models \tau'(\phi)$, then $\CM_2,s\models \phi$, which implies that $\phi$ is also satisfiable.
\end{proof}

\begin{lemma}\label{lem:red-dm2d}
The satisfiability problem of \ldm is polynomial time reducible to that of \ld.
\qed
\end{lemma}

So we finally obtain the following theorem.

\begin{theorem}\label{thm:com-complex}
The satisfiability problem of the eight logics without common knowledge are all PSPACE complete, and the satisfiability problem of the eight logics with common knowledge are all EXPTIME complete.
\end{theorem}
\begin{proof}
Our examination of logics without common knowledge refers to the proof structure outlined in Figure~\ref{fig:sat-nonc}. As the satisfiability problem of K$_n$ ($n \geq 1$, indicating either poly- or multi-modal) is PSPACE complete \cite{HM1992}, so is the satisfiability problem of EL since they share the same axiom system, which is sound and complete for both logics. (Since the complexity of the satisfiability and validity problem remains the same for these logic, the solution to the K$_n$-satisfiability problem can be treated as the solution to the \l-satisfiability problem.)

Similarly, given that the satisfiability problem of K$_n^D$ is PSPACE complete (as mentioned in \cite{HM1992}, though only as a claim, leaving the details to the reader), the satisfiability problem of \ld is also PSPACE complete. Furthermore, as the satisfiability problem of (mono-modal) KB$_1$ is PSPACE complete \cite{CL1994}, the satisfiability problem of \ls is PSPACE hard (it has an axiomatic system that is a multi-modal generalization of that for KB$_1$). As either \l or \ls is a sublogic of the others, namely \lm, \ldm, \lsd, \lsm and \lsdm, all of them have a PSPACE lower bound.

By Lemmas~\ref{lem:red-sdm2dm} and \ref{lem:red-dm2d}, there exist polynomial-time algorithms that reduce the satisfiability problem of \ldm and \lsdm to that of \ld. This implies that the problem is solvable in PSPACE (first execute the reduction algorithm, then call the algorithm for solving \ld which is PSPACE complete). Thus, we conclude that the complexity for the satisfiability problem of all the logics without common knowledge is PSPACE complete.

Turning our attention to the logics with common knowledge, we refer to Figure~\ref{fig:sat-withc} for the proof structure. Similar arguments can be made by noting these known or newly proved results:

(1) The satisfiability problem of K$_n^C$ ($n \geq 1$) and S5$_n^C$ are EXPTIME complete \cite{HM1992}.

(2) The satisfiability problem of \lc can be reduced to that of K$_n^C$ (they share almost the same axiomatic system, but there is a subtlety in that the $E_G$-modality for everyone's knowledge is used as an initial operator, and $E_G$ can be rewritten in terms of individual knowledge only in exponential length, so this reduction only gives us an EXPTIME upper bound).

(3) The satisfiability problem of K$_1^C$ can be reduced to that of \lc in linear time ($E_{\{a\}}$ can be rewritten by $K_a$ in linear time, giving us the EXPTIME hardness).

(4) The satisfiability problem of S5$_n^C$ can be reduced to that of \lsc in polynomial time (Lemma~\ref{lem:red-s52sc}).

(5) The satisfiability problem of \lcdm can be reduced to that of \lc in polynomial time; In a similar vein, the satisfiability problem of \lscdm can be reduced to that of \lsc in polynomial time. (see Lemma~\ref{lem:red-withc}.) 
\end{proof}

\section{Conclusion}

We examined epistemic logics with various types of group knowledge, interpreted over the class of (similarity) models. These models are straightforward extensions of the models used in classical epistemic logic. We delved into their axiomatization and computational complexity results, finding these logics to be notably intriguing.

Although the (similarity) models are generalizations of classical relational models, the logics that exclude mutual knowledge, when interpreted over them, are not entirely new. They parallel classical epistemic logics with common and/or distributed knowledge (based on K and KB frameworks instead of S5, and those based on KB have not been extensively covered in the literature). This could suggest that the classical epistemic languages with group knowledge may lack the expressive power to detail the models fully. This is somewhat hinted at by the fact that the scenario changes when we incorporate mutual knowledge into the language.

Mutual knowledge is a concept of group knowledge that originated from the generalized models. We anticipate that these generalized models can offer more diversity and potential avenues for epistemic logic studies. Simultaneously, the new logics display excellent complexity results for both the model checking problems and satisfiability problems, affirming the notion that they are simple and natural extensions of the classical epistemic logics.

The framework of our logics presents diverse possibilities for characterizing the concept of \emph{knowability}. Apart from interpreting knowability as known after a single announcement \cite{BBvDHHdL2008}, a group announcement \cite{ABDS2010}, or after a group resolves their knowledge \cite{AW2017rdk}, it is now conceivable to perceive knowability as known after an agent acquires certain skills (epistemic abilities) from some source or from a given group. Our framework also enables us to easily characterize \emph{forgetability} or \emph{degeneration} through changes in epistemic abilities, a process that is not as straightforward in classical epistemic logic.

Looking ahead, we aim to explore more sophisticated conditions on the similarity relation, such as those introduced in \cite{CMZ2009}. It would also be valuable to compare our framework with existing ones that use the same style of models, as presented in \cite{NT2015,DLW2021}. This comparative analysis could yield insightful observations and potentially pave the way for further advancements in the field of knowledge representation and reasoning.

\bibliographystyle{plain}
\bibliography{ELW}

\end{document}